%% file: main.tex
\documentclass{article}
\usepackage[utf8]{inputenc}
\usepackage[T1]{fontenc}
\usepackage{lmodern}
\usepackage[a4paper, top=3cm,bottom=3cm,left=3.5cm,right=3.5cm]{geometry}
\usepackage[british]{babel}
\usepackage{csquotes}

\title{Relativistic limits on quantum operations}
\date{May 2021}

\usepackage[style=numeric-comp, 
            sortcites=true,
            giveninits=true, 
            useprefix=true,
            sorting=nyt,
            doi=false,isbn=false,url=false,eprint=false]{biblatex}
\AtEveryBibitem{\clearfield{month}} 
\DeclareFieldFormat{titlecase}{#1} 
\DeclareFieldFormat{postnote}{#1} 
\DeclareSortingNamekeyTemplate{
  \keypart{\namepart{family}}
  \keypart{\namepart{prefix}} 
  \keypart{\namepart{given}}
  \keypart{\namepart{suffix}}
}
\DeclareNameAlias{author}{family-given} 
\renewbibmacro{begentry}{\midsentence}  
\renewbibmacro{in:}{\ifentrytype{article}{}{\printtext{\bibstring{in}\intitlepunct}}} 
\usepackage{tikz, xcolor, wrapfig}
\usepackage[font=small,labelfont=sc,margin=1em]{caption}

\usepackage[pdfusetitle]{hyperref}

\usepackage[shortlabels]{enumitem}
\setlist{noitemsep}

\usepackage{emptypage}
\DeclareFontSeriesDefault[rm]{bf}{b} 

\usepackage{marginnote}
\reversemarginpar
\newif\iftodo
\newif\iftodow
\newif\ifnote
\newif\ifmaybe
\newcommand{\todo}[1]{\iftodo{\marginnote{\normalsize\normalfont\color{blue}\textsf{TODO}}{\sffamily\color{blue}{#1}}}\fi}

\todofalse  
\todowfalse
\notefalse
\maybefalse

\newif\ifnieuw
\newif\iflang 

\nieuwfalse
\langfalse

\newcommand{\ie}{i.e.}
\newcommand{\eg}{e.g.}

\usepackage{amsmath, amsthm, amssymb, bm, commath, extarrows, tikz-cd, accents} 
\newtheorem{theorem}{Theorem}
\newtheorem{corollary}[theorem]{Corollary}
\newtheorem{proposition}[theorem]{Proposition}
\newtheorem{lemma}[theorem]{Lemma}
\theoremstyle{definition}
\newtheorem{definition}[theorem]{Definition}

\usepackage{apptools} 
\AtAppendix{\counterwithin{theorem}{section}} 
\AtAppendix{\counterwithin{equation}{section}}

\newcommand{\defn}[1]{\emph{\textbf{#1}}}

\renewcommand{\a}{\alpha}

\renewcommand{\d}{\delta}
\renewcommand{\phi}{\varphi}
\newcommand{\e}{\varepsilon}
\newcommand{\g}{\gamma}
\newcommand{\s}{\sigma}
\newcommand{\om}{\omega}
\newcommand{\G}{\Gamma}
\newcommand{\ups}{\upsilon}
\newcommand{\tom}{\tilde\omega}
\newcommand{\I}{\mathcal{I}}
\renewcommand{\L}{\Lambda}
\renewcommand{\u}{\upsilon}
\let\o\overline
\let\ul\underline
\let\tns\otimes
\newcommand{\one}{\bm 1}
\newcommand{\ns}{\text{ns}}

\newcommand{\spacelike}{\mathrel{\natural}}
\DeclareMathOperator{\Tr}{Tr}
\DeclareMathOperator{\id}{id}
\DeclareMathOperator{\ch}{ch}

\def\N{\mathbb{N}}

\def\R{\mathbb{R}}
\def\C{\mathbb{C}}

\def\B{\mathcal{B}}
\def\H{\mathcal{H}}
\def\M{\mathcal{M}}

\def\D{D} 
\def\S{S} 
\def\Af{\mathfrak{A}}
\def\Bf{\mathfrak{B}}
\def\Uf{\mathfrak{U}}
\def\Cf{\mathfrak{C}}

\newcommand{\proj}[1]{\left|#1\right\rangle\left\langle#1\right|}
\newcommand{\projk}[1]{|#1\rangle\langle#1|} 
\renewcommand{\k}[1]{\left|#1\right\rangle} 
\newcommand{\kk}[1]{|#1\rangle} 

\newcommand{\black}[1]{\textcolor{black}{#1}}

\usepackage{tikzit}
\input{figures/mystyle.tikzstyle}
\let\klein\scriptsize

\begin{document}

    \input{titlepage}
    
    \newgeometry{top=3cm,bottom=2cm,left=3.5cm,right=3.5cm} 
    
    \noindent\textsc{Abstract.}\quad In this essay we explore the limits imposed by the impossibility of superluminal signalling on the class of physically realisable quantum operations, focusing on the difference in approaches one can take towards this problem in Hilbert space quantum theory and algebraic quantum field theory (AQFT). We review a recent proposal by Fewster and Verch for a generally covariant measurement scheme in AQFT, which we call the FV scheme; along the way, we argue that the framework as set up in its original proposal can be greatly simplified without losing generality or physical motivation. We then discuss the problem, introduced by Sorkin, that naive generalisations of the notion of operation to the relativistic context can lead to signalling between spacelike observers by introducing a third observer. We consider approaches to tackle this problem in Hilbert space quantum theory and argue that these are sometimes ambiguous, since the physical (im)possibility of operations depends in part on the specific spacetime embodiment of the quantum system under consideration, which lies outside the Hilbert space formalism. In the AQFT context, we review and generalise a recent result showing that operations induced by the FV measurement scheme do not enable superluminal signalling. We connect this result to Hilbert space by introducing a concrete model of AQFT, which we call the `hybrid model' as it provides a way to explicitly embed Hilbert space quantum theory into spacetime using the formalism of AQFT. Finally, we show that in this model, under suitable necessary and sufficient conditions, the converse result also holds: every operation that does not enable superluminal signalling can be implemented in the FV framework.
    
    \tableofcontents
    
    \newgeometry{top=3cm,bottom=3cm,left=3.5cm,right=3.5cm}
    
    \newpage
    
    \input{sections/1-introduction}
    \input{sections/2-preliminaries}
    \input{sections/3-fewsterverch}
    \input{sections/4-sorkin}
    \input{sections/5-fv-is-causal}
    \input{sections/6-conclusion}
    \phantomsection
    \addcontentsline{toc}{section}{Acknowledgements}
    \section*{Acknowledgements}
    This essay was written at Wolfson College in the University of Cambridge from March to April 2021 as part of the examination for Part III of the Mathematical Tripos. My greatest thanks go to Jeremy Butterfield, for his devoted supervision and inspiring enthusiasm, and the detailed feedback he provided on earlier versions of this work.

    \appendix
    
    \input{sections/app-isomorphisms}
    \input{sections/app-geometry}
    
    \addcontentsline{toc}{section}{References}
    \printbibliography
\end{document}


%% file: titlepage.tex
\begin{titlepage}
    \newgeometry{top=4cm,bottom=4cm,left=2cm,right=2cm,marginparwidth=1.75cm}
    
    \center
    \large

    \includegraphics[width=.3\textwidth]{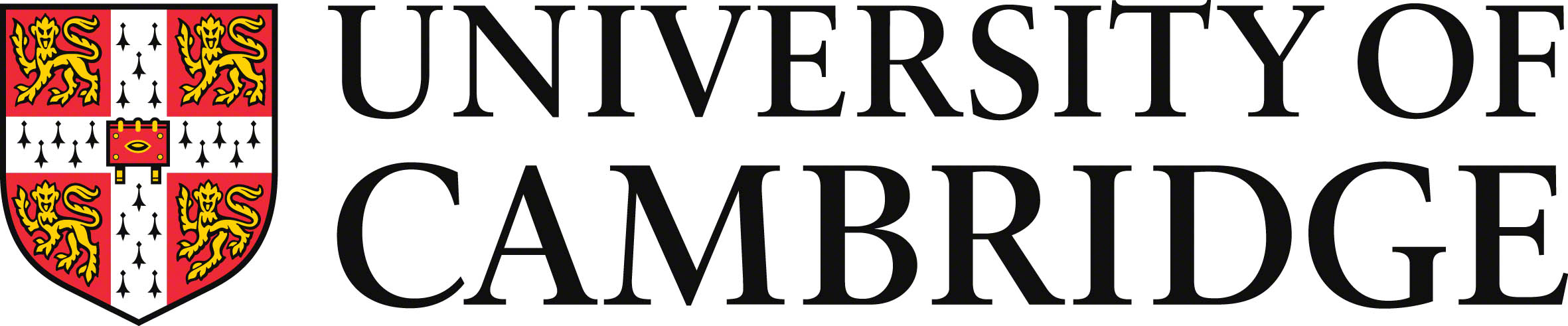}\\[2cm]
    
    \makeatletter
    {\sffamily\huge Relativistic limits on quantum operations}\\[3cm]
    \makeatother
    
    {\Large\textsc{Part III Essay}}\\[3cm]
    
    {\Large \textsc{Tein van der Lugt}}\\[0.5cm]
    
    Department of Applied Mathematics and Theoretical Physics\\
    Wolfson College\\
    University of Cambridge\\ \vfill
    
    \emph{Essay setter:}\\[0.3cm]
    Dr \textsc{Jeremy Butterfield}\\ \vfill
    \makeatletter
    {\large \@date}
    \makeatother


    \vfill
\end{titlepage}

%% file: sections/1-introduction.tex
\section{Introduction}
When discussing quantum operations in the general context of quantum information theory, one often considers the set of completely positive, trace-preserving (CPTP) maps between bounded operators on finite-dimensional Hilbert spaces. While complete positivity and trace-preservation are necessary conditions for such a map to represent a physically possible transformation between quantum states, there are additional restrictions, both in practice and in principle, which physically realisable operations must obey. An important example is that they should not lead to superluminal signalling, an in principle restriction arising from relativity theory and which will be the subject of this essay.

The problem of reconciling quantum mechanics with relativity theory is of course a notoriously difficult one. Nonrelativistic quantum mechanics benefits from the operationally well-established (though foundationally insecure) concepts of ideal measurement and the instantaneous state update, or wavefunction collapse, induced by it. In relativistic spacetime, however, instantaneity is relative, leading to questions about the ontology of state update. But problems also arise in a purely operational\footnote{Operationalism can be formulated as the view that a concept is only meaningful when one has a method of measurement or detection for it; more abstractly, it views any concept as nothing more than a ‘set of operations’ (here not necessarily to be interpreted in the quantum mechanical sense) \cite{sep-operationalism}. The operationalist view is prominent in many areas of quantum mechanics including algebraic quantum field theory and quantum information theory.} context: it turns out that when naively generalising ideal measurements (and more generally, operations) to relativistic spacetime, allowing arbitrary operations to be performed by local parties can lead to superluminal signalling, as pointed out by \textcite{Sor93}.

This suggests the general problem of delineating the class of quantum operations that correspond to a physically realisable change of state,%
\footnote{That is, physically realisable in the operational sense that it can be induced by an outside observer. For the purposes of this essay, we regard any physical process that is not due to the intervention of an outside agent as a unitary evolution rather than an operation. In the (about to be introduced) framework of algebraic quantum field theory, this unitary evolution is encoded in one quantum state defined on the entirety of spacetime, and hence does not correspond to a ‘change of state’. Moreover, note that we call this a `class' rather than a set, since `physically realisable' is not (yet) a rigorous notion and there might be borderline cases in which the physical realisability of an operation depends on factors lying outside of the present-day theories of physics.}\textsuperscript{,}%
\footnote{A different though very closely related question which has received more attention \cite{Zan01,Nie97,Wig52,AY60,Dir58} is whether all observables of a quantum system are truly measurable. Compare \textcite[my translation]{Wig52}: ``\textit{The usual assumption of the statistical interpretation of quantum mechanics that all Hermitian operators represent measurable quantities is, arguably, generally recognised as a convenient mathematical idealisation and not as an expression of fact.}''
with \textcite[p.~37]{Dir58}: ``\textit{The question now presents itself---can every observable be measured? The answer theoretically is yes. In practice it may be very awkward, or perhaps even beyond the ingenuity of the experimenter, to devise an apparatus which could measure some particular observable, but the theory always allows one to imagine that the measurement can be made.}''}
and that in particular do not lead to the possibility of superluminal signalling. Many different approaches to this problem can be and have been taken. One way of comparing these it to distinguish between ‘bottom-up’ and ‘top-down’ approaches. Here, bottom-up approaches attempt to solve the problem by considering the specific physical mechanisms by which operations can be implemented, either exhibiting protocols for implementing particular (classes of) operations or constructing a general framework for implementing operations that (hopefully) capture all realisable operations \cite{BFR21, HK69, HK70a, OO16}.

In a top-down approach, on the other hand, one starts from a predefined set of quantum operations (\eg\ CPTP maps), and by considering physical constraints such as the impossibility of superluminal signalling aims to either show on a case-by-case basis that some operations violate these constraints, or derive general necessary and/or sufficient conditions for operations to satisfy the constraints---without giving procedures to actually perform those operations that do \cite{Sor93, BGNP01, ESW02}. In other words, a top-down approach, if you will, carves out the class of physically possible operations from the block of marble formed by \eg\ CPTP maps.

Another, related distinction between approaches to the aforementioned problem is what mathematical framework is used to describe the quantum systems. Quantum mechanics on finite-dimensional Hilbert space is the most widely used language when talking about measurements and operations, and the issue of superluminal signalling has been researched extensively in this framework in the top-down approach \cite{Sor93, BGNP01, ESW02}. Locality is then encoded by a tensor factorisation of the Hilbert space into multiple systems controlled by different observers. An aspect that makes this approach somewhat ambiguous, however, is the fact that spacetime structure itself falls outside the scope of the Hilbert space formalism and must instead be considered separately and on a case-by-case basis. Indeed, in Section~\ref{sec:bgnp} of this essay, focusing on the purely Hilbert space treatment of the problem by \textcite{BGNP01}, I will argue that whether or not a Hilbert space operation enables superluminal signalling can non-trivially depend on the specific physical realisation of the quantum systems described by the tensor product factors.

Bottom-up approaches, on the other hand, involve giving a description of the true physical mechanism which goes into performing the operations, and so are best formulated in frameworks that are more complete and take relativity as fundamental. These are of course quantum field theories. The main focus of this essay will be on a generally covariant measurement scheme in the rigorous language of algebraic quantum field theory (AQFT), recently proposed by \textcite{FV20}, which can be regarded as a bottom-up approach to the problem of superluminal signalling. Although the measurement scheme was primarily designed to consistently describe measurements and their induced state updates in AQFT, it can also be argued to more generally capture all physically realisable operations. The central observation of this essay will indeed be that operations induced by this scheme do not enable superluminal signalling, as was recently shown by \textcite{BFR21}.

To investigate the connection between the approaches in Hilbert space quantum mechanics \cite{BGNP01} and AQFT \cite{BFR21}, we introduce a simple model of AQFT which embeds, as it were, the formalism of nonrelativistic Hilbert space QM into relativistic AQFT---therefore we will refer to it as the `hybrid' model. Apart from providing intuition behind the result of \textcite{BFR21}, its hybrid nature allows us to prove that in this particular model and under a mild geometrical condition, the converse of that result also holds: all operations that do not enable superluminal signalling can be implemented in the FV framework. This is achieved by carrying over the important result from Hilbert space QM that `all semicausal operations are semilocalisable' \cite{BGNP01,ESW02} to the hybrid model of AQFT.

We start with some preliminaries in Section~\ref{sec:preliminaries}, including a brief introduction to algebraic quantum field theory and the (relatively little-studied) notions of (local) operations in AQFT. We introduce the hybrid model in Section~\ref{sec:prelim-aqft-hybrid-model}. In Section~\ref{sec:fewster-verch}, we review the generally covariant measurement scheme of \textcite{FV20}; along the way, we argue that the framework as set up in its original proposal can be greatly simplified without losing generality (which we prove in Appendix~\ref{app:semicolons}). Then, in Section~\ref{sec:sorkin}, we zoom in on Sorkin’s observation that arbitrary operations can establish superluminal signalling, and discuss the approach in Hilbert space QM as set out in \cite{BGNP01}. In Section~\ref{sec:bfr} we review the result of \cite{BFR21}, generalising it slightly to show that operations induced by Fewster and Verch’s measurement scheme do not enable superluminal signalling. Finally, Section~\ref{sec:bfr-hybrid-model} applies these ideas to the hybrid model to arrive at the result discussed above.

%% file: sections/2-preliminaries.tex
\section{Preliminaries}\label{sec:preliminaries}
This section contains some background and terminology necessary to set the stage for the rest of the essay. We begin in Section~\ref{sec:prelim-functional-analysis} with some basic definitions in the context of *-algebras. In Section~\ref{sec:prelim-qm-hilbert-space}, we recall some aspects of Hilbert space quantum mechanics, which I mostly assume as familiar; we will focus on the connection with the algebraic approach to quantum mechanics, which will be useful for comparing it to algebraic quantum field theory. Section~\ref{sec:prelim-lorentz} introduces the necessary terminology and results from Lorentzian geometry, which we will need before continuing on to AQFT in Section~\ref{sec:prelim-aqft}. Although the essay is set in the context of general curved spacetime, an understanding of flat (Minkowski) spacetime will generally be sufficient.

\subsection{Functional analysis}\label{sec:prelim-functional-analysis}
In algebraic approaches to quantum mechanics, *-algebras play an important role as generalisations of the set of operators on a Hilbert space, hence representing the observables of a quantum system. A \defn{*-algebra} $\Af$ \footnote{That's an `A', by the way.} is an (associative but generally noncommutative) algebra over $\C$ together with an \defn{involution} $^*: \Af\to\Af$: that is, for each $a,b\in\Af$ and $\lambda\in\C$ we have
\begin{itemize}
    \item $a^{**} = a$;
    \item $(ab)^* = b^*a^*$; and
    \item $(\lambda a)^* = \o\lambda a^*$.
\end{itemize}
$\Af$ is called \defn{unital} if it has a unit, denoted $\one_\Af$. Given two unital *-algebras $\Af$ and $\Bf$, one can take the tensor product $\Af\tns\Bf$; this is again a unital *-algebra with involution $(a\tns b)^* := a^* \tns b^*$ for $a\in\Af$ and $b\in\Bf$, extended linearly. Moreover, a \defn{*-isomorphism} $\phi:\Af\to\Bf$ between *-algebras is a bijection that commutes with the algebraic operations and involutions on $\Af$ and $\Bf$, and a \defn{*-subalgebra} is a subset of $\Af$ which is closed under algebraic operations and the involution.


We say that an element $a$ of a *-algebra $\Af$ is \defn{positive} if $a=b^*b$ for some $b\in\Af$. This implies that $a$ is \defn{self-adjoint}, \ie\ $a^*=a$. This notion extends to a partial order $\leq$ on $\Af$ by setting $a \leq b$ iff $b-a$ is positive. In particular, if $\Af$ is unital, we say that $a\in\Af$ is an \defn{effect} if $0\leq a\leq \one_\Af$.


On the other hand, a linear \emph{map} $T:\Af\to\Bf$ between *-algebras is said to be \defn{positive} if it preserves positivity, \ie\ for all $a\in\Af$ there is a $b\in\Bf$ such that $T(a^*a) = b^*b$. Moreover, it is called \defn{completely positive} (CP) if for any other *-algebra $\Bf$, the extension $T\tns\id_{\Bf}: \Af\tns\Bf\to\Af\tns\Bf$ is positive. Here $\id_{\Bf}: \Bf\to\Bf$ is the identity map. Finally, if $\Af$ and $\Bf$ are unital then $T$ is also called \defn{unital} if $T(\one_A) = \one_B$.

An important example of a positive map is a \defn{state} on a unital *-algebra $\Af$, which is defined as a bounded linear functional $\om:\Af\to\C$ being
\begin{itemize}
    \item positive: $\om(a^* a) \geq 0$ for all $a\in\Af$,\footnote{Note that $\C$ is a *-algebra with complex conjugation as the involution and that $a\in\C$ is positive iff $a\in\R$ and $a\geq 0$.} and
    \item \defn{normalised}: $\om(\one_\Af) = 1$.
\end{itemize}
We denote the set of all states on $\Af$, the \defn{state space}, by $\S(\Af)$; this is a convex subset of the dual $\Af^*$. States can be seen as assigning an expectation value $\om(a)\in\C$ to each observable $a\in\Af$, thereby specifying a normalised state of the quantum system with observables $\Af$.\footnote{\label{fn:sa-observables}Often only self-adjoint elements of the algebra (\eg\ self-adjoint operators on a Hilbert space) are considered to be observable, and their expectation values are required to be real numbers, so that a `state' should be defined as a positive and normalised map $\Af_{\text{sa}}\to\R$, where $\Af_{\text{sa}}\subseteq \Af$ is the subalgebra of self-adjoint elements. Denote the space of such states by $\S_\text{sa}(\Af)$. If $\Af$ is a \emph{C*-algebra}, however, which is nearly always assumed in physical scenarios, one can write any element of $\Af$ as a linear combination of positive elements \cite[Lemma~C.53]{Lan17}. This implies that $\S(\Af)$ and $\S_\text{sa}(\Af)$ are in bijective correspondence (cf.\ \textcite[Proposition~2.6]{Lan17}).}

Going one step further yet, we say that a map $\G:\Af^*\to\Bf^*$ is \defn{positive} if it preserves positivity of functionals $\om\in\Af^*$, and \defn{completely positive (CP)} if its extension $\G\tns\id_{\Cf^*}$ is positive in this sense for any *-algebra $\Cf$.

The example of a *-algebra which is most relevant to us is of course the space $\B(\H)$ of bounded operators on an arbitrary Hilbert space $\H$, where the algebraic operations are as usual, the norm is the operator norm and the involution is given by the adjoint (which we still denote as $^*$). $\B(\H)$ is unital with unit $\one_{\B(\H)} =: \one_\H$, the identity operator on $\H$. Any element $a\in\B(\H)$ is positive in the sense defined above iff for all $\psi\in\H$, $\langle \psi, a\psi\rangle \geq 0$---that is, $a$ is positive semidefinite as an operator.
In particular, an effect $0\leq a\leq \one_\H$ uniquely defines a two-element POVM $\{a, \one_\H-a\}$, \ie\ a two-outcome measurement or `yes/no' question about the quantum system.
As for states $\om\in\S(\B(\H))$, in the finite-dimensional case they correspond bijectively to density operators, as we will see shortly.

    

\subsection{Quantum mechanics on Hilbert space}\label{sec:prelim-qm-hilbert-space}
We recall the notion of operation in the case of quantum systems on a Hilbert space $\H$ and explain the relation between density operators $\D(\H)$ and states $\S(\B(\H))$. We will see that this relation is closely linked to the idea of the equivalence between the Schrödinger and Heisenberg pictures of quantum evolution. In this section, we restrict ourselves to finite-dimensional Hilbert spaces.

\begin{definition}\label{def:hs-operation}
    Let $\H_1$ and $\H_2$ be finite-dimensional. An \defn{operation} is a map $\L:\B(\H_1)\to\B(\H_2)$ which is linear, completely positive (CP) and trace-nonincreasing, \ie\ $\Tr(\L(\rho)) \leq \Tr(\rho)$ for $\rho\in\D(\H_1)$. We call $\L$ \defn{non-selective} or CPTP if it is trace-preserving, and \defn{selective} otherwise. If it is non-selective, it can be seen as a map between the spaces of density operators $\D(\H_1)$ and $\D(\H_2)$.
\end{definition}

The central equivalence linking the Hilbert space and algebraic approaches to quantum mechanics is given by the following proposition (cf.\ \textcite[Thm.~2.8]{Lan17}; see Thm.~4.12 of that work for the infinite-dimensional case).

\begin{proposition}\label{prop:dh-sbh}
    Let $\H$ be finite-dimensional. $\B(\H)$ and its dual $\B(\H)^*$ are isomorphic as vector spaces via
    \begin{equation}
        \om(a) := \Tr(\rho a) \text{\quad for all } a\in\B(\H),
    \end{equation}
    defining a functional $\om\in\B(\H)^*$ for a given operator $\rho\in\B(\H)$. Under this equivalence, $\om$ is positive iff $\rho$ is positive and $\om(\one_{\H}) = 1$ iff $\Tr(\rho) = 1$, meaning that $\D(\H)$ and $\S(\B(\H))$ are isomorphic as compact convex sets.
    
\end{proposition}

The following two remarks about this equivalence will be useful to us.
\begin{enumerate}[(i)]
    \item Let $\H_1$ and $\H_2$ be two finite-dimensional Hilbert spaces. Then under the equivalence of Theorem~\ref{prop:dh-sbh}, partial trace $\Tr_2 : \D(\H_1\tns\H_2) \to \D(\H_1)$ on density operators corresponds to the mapping $\S(\B(\H_1\tns\H_2)) \ni \om_{12} \mapsto \om_1 \in \S(\B(\H_1))$, where
    \begin{equation}
        \om_1(a) = \om_{12}(a\tns\one_{\H_2})  \text{\quad for } a\in\B(\H_1).
    \end{equation}
    That is, $\om_1$ is simply the restriction of $\om_{12}$ to $\B(\H_1)$ under the isomorphism $\B(\H_1) \cong \B(\H_1)\tns\{\one_{\H_2}\}$.
    
    \item Let $\L:\B(\H_1)\to\B(\H_2)$ be a positive map and $\rho\in\B(\H_1)$ a positive operator. Then for $a\in\B(\H_2)$,
    \begin{equation}\label{eq:tr-hilbert-schmidt}
        \Tr(\L(\rho)a) = \Tr(\L(\rho)^* a) = \Tr(\rho^* \L^\dagger(a)) = \Tr(\rho\L^\dagger(a)),
    \end{equation}
    where $\L^\dagger$ is the adjoint of $\L$ with respect to the Hilbert-Schmidt inner products on $\B(\H_1)$ and $\B(\H_2)$. So under the equivalence of Proposition~\ref{prop:dh-sbh}, $\L:\B(\H_1)\to\B(\H_2)$ corresponds to the pullback\footnote{Also often called adjoint, but we have enough of those already.} of $\L^\dagger$, which we denote by $\widehat{\L^\dagger} : \B(\H_2)^* \to \B(\H_1)^*$, \ie\ $(\widehat{\L^\dagger}(\om))(a) = \om(\L^\dagger a)$ for $a\in\B(\H_2)$. Moreover, it can be shown by some straightforward arguments that \cite{ESW02}
    \begin{itemize}
        \item $\L$ is CP iff $\L^\dagger$ is CP iff $\widehat{\L^\dagger}$ is CP;
        \item $\L$ is trace-nonincreasing iff $\L^\dagger(\one_{\H_2}) \leq \one_{\H_1}$ iff $\widehat{\L^\dagger}$ is normalisation-nonincreasing; and
        \item $\L$ is trace-preserving iff $\L^\dagger$ is unital iff $\widehat{\L^\dagger}$ is normalisation-preserving.
    \end{itemize}
    
    Although the Schrödinger and Heisenberg pictures in quantum mechanics usually refer to the unitary evolution of a quantum system, the same idea applies to general operations: $\L$ and $\widehat{\L^\dagger}$ represent the operation in the Schrödinger picture, where quantum states are the entities that transform, while $\L^\dagger$ is the same operation expressed in the Heisenberg picture, where observables are the entities that transform.
    
    Finally, if $\Af=\B(\H)$ with $\H$ finite-dimensional and $\G:\Af^*\to\Af^*$ is any CP and normalisation-nonincreasing map, then Proposition~\ref{prop:dh-sbh} and the results above tell us that there exists a map $\L^\dagger: \Af\to\Af$ such that $\G$ is the pullback of $\L^\dagger$. This is not true for general unital *-algebras $\Af$---we will come back to this in Section~\ref{sec:prelim-aqft-operations-locality}.
\end{enumerate}

Proposition~\ref{prop:dh-sbh} is a window into the algebraic approach to quantum mechanics, where the focus is on general algebras of observables rather than wavefunctions living in Hilbert space. We will encounter the algebraic approach in greater detail when discussing algebraic quantum field theory in Section~\ref{sec:prelim-aqft}.

\subsection{Lorentzian geometry}\label{sec:prelim-lorentz}
Algebraic quantum field theory takes place on a time-oriented Lorentzian manifold $M$, that is, a smooth manifold of dimension at least two equipped with a Lorentzian metric $g$ of signature $-++\cdots+$ and a choice of time orientation---\ie\ a smooth timelike vector field designating the `future' time direction at each point, such that every causal vector is future- or past-directed. Here a vector is \defn{causal} if it is timelike or null; moreover, a curve in $M$ is causal if its tangent is everywhere causal.

For a point $p\in M$, the \defn{future (past) lightcone} or \defn{causal future (past)} of $p$ is the set of all points that can be reached from $p$ by a piecewise smooth, 
future- (past-)directed causal curve. These sets are denoted by $J^+(p)$ and $J^-(p)$, respectively. These notions naturally extend to sets, which are our main concern in AQFT: for $O\subseteq M$, define $J^\pm(O) = \bigcup_{p\in O} J^\pm(p)$. The \defn{causal hull} of a set $O\subseteq M$ is the intersection $J^+(O) \cap J^-(O) =: \ch(O)$. $O$ is called \defn{causally convex} if it is equal to its causal hull, or equivalently, if it contains every piecewise smooth causal curve that starts and ends in $O$ \cite{FV20}. This last property can be easily seen to imply that the intersection of two causally convex sets is again causally convex. (See \eg\ Figure~\ref{fig:smiley-face} on page~\pageref{fig:smiley-face} for an example of what a causally convex region can look like.) The \defn{causal complement} of $O$, denoted by $O^\perp$, is the set $M\setminus(J^+(O)\cup J^-(O))$; also define
\begin{equation}
    M_O^+ := M\setminus J^-(O) \text{\qquad and \qquad} M_O^- := M\setminus J^+(O).
\end{equation}
Then we have the following result, which we will use later and will prove here for illustration.
\begin{lemma}\label{lemma:ccompl-is-cconvex}
    For any set $O\subseteq M$, the sets $M_O^+$, $M_O^-$ and $O^\perp$ are causally convex.
\end{lemma}
\begin{proof}
    Suppose that $M_O^+$ is not causally convex, so that there exist $p,q\in M_O^+$ and a future-directed causal curve from $p$ to $q$ that intersects $J^-(O)$, say in the point $r$. Then by definition of $J^-(O)$, there is a future-directed causal curve from $r$ into $O$, so that by joining this with the curve from $p$ to $r$ obtained by restricting the curve from $p$ to $q$, we see that $p\in J^-(O)$, yielding a contradiction. Similarly, $M_O^-$ is causally convex. Since any intersection of causally convex sets is again causally convex, $O^\perp = M_O^+ \cap M_O^-$ is also causally convex.
\end{proof}

A point $p\in M$ is a \defn{past} or \defn{future endpoint} of a smooth causal curve $\g:(a,b)\to M$ with $a,b\in\R$ if $\g(t) \to p$ as $t\to a$ or $t\to b$, respectively. The curve is \defn{past-/future- inextendible} if it has no past/future endpoint, and simply \defn{inextendible} if it has neither.
A \defn{Cauchy surface}, then, is a set $\Sigma\subseteq M$ such that every inextendible smooth timelike curve intersects $\Sigma$ exactly once. $M$ is \defn{globally hyperbolic} iff a Cauchy surface exists \cite{Ger70}; if it is, then $M$ may be foliated by everywhere spacelike Cauchy surfaces, which may therefore be seen as `time slices' \cite{FV20} (see also the proof of Lemma~\ref{lemma:separating-cauchy-surface} below). A causally convex open submanifold of a globally hyperbolic spacetime $M$ is again globally hyperbolic \cite{FV12}.
The \defn{domain of dependence} or \textit{Cauchy development} of a set $O\subseteq M$ is $D(O) := D^+(O)\cup D^-(O)$, where $D^{+/-}(O)$ is the set of points $p$ such that every past-/future-inextendible smooth causal curve through $p$ intersects $O$. For example, if $O'\subseteq O$ contains a Cauchy surface of $O$, then $D(O')=D(O)$.
$D(O)$ is often `diamond-shaped' (see \eg\ Figure~\ref{fig:D-O-covering} on page~\ref{fig:D-O-covering}) and should be thought of as the set of points $p$ such that the `initial data' in $O$ uniquely determine system at $p$ \cite{Ger70}.

Finally, we introduce some causal relations between subsets of spacetime which will be of significance to us. We write 
\begin{equation}
    O_1 \spacelike O_2 \quad:\iff\quad O_1 \subseteq O_2^\perp \quad (\iff O_2 \subseteq O_1^\perp)
\end{equation}
for $O_{1,2}\subseteq M$, in which case we say that $O_1$ and $O_2$ are \defn{causally disjoint} or \defn{spacelike separated}.
Moreover, for any finite set $\mathcal O$ of precompact subsets of spacetime $M$ (\ie\ sets with compact closure), we say that a binary relation $\leq$ on $\mathcal O$ is a \defn{causal order} if it is a linear order such that for any $O, P\in\mathcal O$,
\begin{equation}\label{eq:causal-order}
    O < P  \quad\implies\quad  \bar P \subseteq M^+_{\bar O},
\end{equation}
where the bar denotes topological closure. A motivation for this definition is the fact that if $\bar P \subseteq M^+_{\bar O}$ holds for precompact $O$ and $P$, then there is a Cauchy surface (`time slice') separating $O$ and $P$, with $O$ to its past and $P$ to its future, as shown by Lemma~\ref{lemma:separating-cauchy-surface} on page~\pageref{lemma:separating-cauchy-surface}. Moreover, in flat spacetime, we would be able to say that there is an inertial observer according to which $O$ happens completely before $P$.

Note that a causal order as defined in Eq.~\eqref{eq:causal-order} does not always exist (\eg\ if two or more of the regions in the set $\mathcal O$ overlap), and that if it exists, it is not necessarily unique. If $\bar O_1$ and $\bar O_2$ are causally disjoint, for example, then both $O_1\leq O_2$ and $O_2\leq O_1$ constitute causal orders on the set $\{O_1,O_2\}$. If, on the other hand, $\{O_i\}_{i=1}^n$ is a set of $n$ distinct subsets of spacetime such that for all $1\leq i\leq n-1$ we have that $\bar O_{i+1}\subseteq M^+_{\bar O_i}$ and $O_i$ and $O_{i+1}$ are \emph{not} causally disjoint, then there exists a unique causal order, given by ${O_i\leq O_j} \iff i\leq j$ (since $O_{i+1} < O_i$ cannot hold for any causal order on $\{O_i\}_{i=1}^n$ by \eqref{eq:causal-order}, and a causal order is linear, so that $O_i < O_{i+1}$ must hold).%
\footnote{The definition of causal order in Eq.~\eqref{eq:causal-order} is admittedly a bit awkward because it requires one to switch between precompact sets and their closures. We have made this choice because Lemma~\ref{lemma:separating-cauchy-surface} needs the sets $K,L$ to be compact, whereas locality in AQFT is formulated in terms of open (causally convex) subsets of spacetime (see Section~\ref{sec:prelim-aqft-axioms}). Note that this makes our definition of causal order slightly different from a similar one considered in \cite{BFR21}.}

\subsection{Algebraic quantum field theory}\label{sec:prelim-aqft}
\input{sections/2A-aqft}

%% file: sections/2A-aqft.tex
This section outlines the fundamentals of algebraic quantum field theory (AQFT). I will start by placing it into context. More comprehensive introductions can be found in, for example, \cite{FR20,Hal07,FV15,Lan96}.

AQFT, introduced by \citeauthor{HK64} in \citeyear{HK64} \cite{HK64}, is one of the three main approaches to quantum field theory, together with the Wightman formulation and traditional, Lagrangian QFT \cite{Swa17}. While Lagrangian QFT is by far the most successful of the three in terms of predictive power and empirical success, it also has many problems, including its lack of mathematical consistency and its reliance on non-local concepts such as the vacuum and the path integral \cite{FR15}.

AQFT is, on the other hand, a mathematically rigorous, axiomatic approach based from the outset on locality principles; it is for this reason that it also goes by the name of `local quantum physics'. One of the main motivations for an axiomatic approach to QFT is that it allows one to attempt to answer to the question `what counts as a model of QFT' by separating general results from results about specific models. Another is that, because the axioms reflect basic physical principles, the physical significance of results is often clearer \cite{Swa17}. AQFT can be formulated on curved spacetimes (which are, however, fixed in advance: that is, one does usually not consider the backreaction of the quantum fields on the spacetime metric); this has made into a useful tool to study general problems involving quantum fields in curved spacetimes such as the Hawking effect and cosmological problems \cite{FR15}.

On the other hand, partly due to its general and axiomatic approach, it is difficult to construct interesting explicit models of the AQFT axioms. No interacting models in 1+3 dimensions are yet under control in AQFT (nor in other rigorous approaches such as Wightman's), including important physical models like QCD, QED or the standard model itself \cite{FR15}.
The generality and rigour of AQFT is useful to us, however, because it enables one to discuss general measurement processes, as we will do in Section~\ref{sec:fewster-verch}.

Section~\ref{sec:prelim-aqft-axioms} below outlines the general framework and main axioms of AQFT, while in Section~\ref{sec:prelim-aqft-operations-locality} we discuss its notions of operation and locality of operations. In Section~\ref{sec:prelim-aqft-hybrid-model}, I introduce a concrete model of AQFT which will turn out useful to us later.



\subsubsection{The framework and axioms}\label{sec:prelim-aqft-axioms}
We start from a time-oriented, globally hyperbolic Lorentzian manifold $M$;\footnote{Another type of approach to AQFT does not assume the existence of one fixed, `big' spacetime; see Appendix~\ref{app:semicolons}.} typically one also assumes that $M$ has at most finitely many connected components \cite{FV20, Few20, FV15}. We define a \defn{(local) region} of $M$ to be any subset $O\subseteq M$ which is open and causally convex.

Just like the algebraic approach to nonrelativistic quantum mechanics, the focus of AQFT lies in algebras of observables.\footnote{See footnote~\ref{fn:sa-observables}. Also, given a state, one can construct a representation of the observables as bounded operators on a Hilbert space via the GNS construction. See \eg\ \cite{FR20}; we will not need this construction here.} Locality is then baked into the formalism by associating to each region $O$ in spacetime a subalgebra which can be viewed as containing those observables that correspond to the part of the system lying in $O$, and can therefore in principle(!)\ be measured by a procedure confined to $O$.\footnote{We will see that this interpretation should not be taken too literally, since the ideal measurement of some observables can lead to superluminal signalling; we will discuss in Section~\ref{sec:fewster-verch} what actual `measurement' can be taken to mean.}

Concretely, an AQFT on $M$ is characterised by a unital *-algebra $\Af(M)$, also notated as $\Af$ and referred to as the \defn{(global) observable algebra}, together with an assignment $O\mapsto \Af(O)$ mapping regions $O$ to *-subalgebras $\Af(O)$ of $\Af(M)$ which share the unit of $\Af(M)$. We refer to $\Af(O)$ as the \defn{local observable algebra} belonging to $O$, and an observable $a\in\Af(O)$ is said to be \defn{localisable in $O$}.
From the point of view of regular QFT, $\Af(O)$ could, for example, consist of algebraic combinations of smeared fields
\begin{equation}\label{eq:smear}
    \int_M f(x)\phi(x)\dif x,
\end{equation}
where $\phi$ is a quantum field and $f$ a function that vanishes outside $O$.

Let us now review some of the basic requirements often made of the assignment $O\mapsto \Af(O)$.\footnote{There are many variants of AQFT which consider different extensions of this set of requirements. In the special case of Minkowski spacetime, for example, the additional symmetry allows one to impose further axioms expressing, for example, that the net of local algebras `plays nicely' with the Minkowski spacetime symmetries (\textit{Poincaré covariance}) or that energy is positive (the \textit{spectrum condition}). This is often referred to as the \textit{Haag-Kastler framework}; see \eg\ \textcite{FR15}.} The first is straightforward and requires no explanation:
\begin{description}
    \item[Isotony.] For two regions $O_1, O_2\subseteq M$ such that $O_1\subseteq O_2$, we have
    \begin{equation}
        \Af(O_1) \subseteq \Af(O_2).
    \end{equation}
\end{description}

The second axiom encodes the causal independence of spacelike separated regions, motivated by relativity theory. Observables in spacelike separated regions should be simultaneously measurable and hence should commute:
\begin{description}
    \item[Einstein Causality.] \label{einstein-causality}For two regions $O_1, O_2\subseteq M$ such that $O_1\spacelike O_2$, we have
    \begin{equation}
        \Af(O_1) \subseteq \Af(O_2)'
    \end{equation}
    and vice versa, where $\Af(O_2)'$ is the commutant of $\Af(O_2)$.
\end{description}
An alternative and often used motivation for Einstein Causality is that at least in some cases, it rules out superluminal signalling. `Locality' of an operation $T:\Af\to\Af$ to a region $O_1$ is sometimes taken to mean that $T$ has a Kraus representation $T(a) = \sum_k c_k^* a c_k$ with $\sum_k c_k^* c_k = \one_\Af$, where the Kraus operators $c_k$ are in $\Af(O_1)$ (in the case of a projective measurement, for example, they are given by the projections onto the eigenspaces of an observable localised in $O_1$). Einstein causality then implies that $T(b) = b$ for $b\in\Af(O_2)$, so that operations $T$ local to $O_1$ do not signal to $O_2$. This is roughly the content of the quantum no-signalling theorem \cite{GRW80}. However, in AQFT, it is debatable whether Einstein Causality is sufficient to completely rule out superluminal signalling with single operations \cite{RV10}. In addition, as we will see in Section~\ref{sec:sorkin}, using multiple operations local to different regions can in general establish superluminal signalling.
We will return to the notion of local operations in AQFT more formally in Section~\ref{sec:prelim-aqft-operations-locality}, where we use a more general definition than the above.

Our third axiom is:
\begin{description}
    \item[Diamond Axiom.] If $O\subseteq M$ is any region, then \cite{But07}
    \begin{equation}\label{eq:diamond-axiom}
        \Af(D(O)) = \Af(O).
    \end{equation}
\end{description}
Intuitively, this encodes the existence of dynamics. Indeed, the physics of the system in the domain of dependence $\D(O)$ (which, at least in flat spacetime, can be seen as `diamond-shaped') should be completely determined by the physics in $O$ \cite{Ger70};\footnote{Note that actual measurements and state updates, which form the indeterministic component of quantum mechanics, are so far completely out of the picture. When one speaks of a self-adjoint algebra element being an in principle measurable physical quantity, these measurements are only hypothetical.} therefore, any observable measurable by a procedure in $D(O)$ should also be measurable by an (albeit probably different) procedure in $O$. (An equivalent formulation often used and referred to as the \emph{time slice property} states that if $O_1$ contains a Cauchy surface (`time slice') of $O_2$, then $\Af(O_2) \subseteq \Af(O_1)$ \cite{Few20}.)

The last axiom relevant for our discussion is the Haag property.
\begin{description}
    \item[Haag Property.] If $K\subseteq M$ is compact and $a\in\Af(M)$ commutes with all elements of $\Af(K^\perp)$,\footnote{\label{fn:K-perp-is-region}Note that $K^\perp$ is a region, so $\Af(K^\perp)$ is indeed defined: $K$ is compact and thus closed, meaning $K^\perp$ is open \cite[Lemma~A.4]{FV12}, and $K^\perp$ is causally convex by Lemma~\ref{lemma:ccompl-is-cconvex}.} then $a\in\Af(L)$ for any connected region $L\supset K$.
\end{description}
This can be proven to hold in specific models, but no proof is known in the general *-algebraic context \cite{FV20}.\footnote{\label{fn:haag-duality}Formulations of this property vary slightly between authors, and sometimes the stronger \emph{Haag duality} is assumed (see Section~\ref{sec:prelim-aqft-hybrid-model}). We follow \textcite{FV20} in our formulation of the Haag property.}

Note that due to Isotony and the Diamond Axiom, a single observable of $\Af(M)$ can be located in many, possibly disjoint, regions.

While the local observable algebras of AQFT encode the degrees of freedom and the dynamics of the system, the physical \emph{state} itself is specified by assigning an expectation value to all observables, and hence by a state $\om$ on $\Af(M)$, as per the definition in Section~\ref{sec:prelim-functional-analysis}. We see that the Heisenberg picture of the unitary evolution of isolated systems is inherent to AQFT: the single quantum state $\om$ is relevant at every spacetime point, while the observables are subject to time dynamics, which is not specified but implicitly assumed by the Diamond Axiom.

Before we continue, it is worth mentioning that multiple different conventions and notations for AQFT exist; in particular, two papers \cite{FV20, Few20} which are of considerable significance to us employ a slightly less straightforward notation for the local observable algebras, which seems to be more general than the notation employed here (and in many other treatments of AQFT \cite{BFR21, FR20, Hal07}). I explain in Appendix~\ref{app:semicolons} why at least for our purposes, however, this notation is not any more general; moreover, our notation greatly simplifies the treatment of the scattering morphism and the proofs of Proposition~\ref{prop:fv-is-local} and Proposition~\ref{prop:FV-local-2} in Section~\ref{sec:fewster-verch}. This difference in notation also explains why, if one compares the above axioms to the ones given by \eg\ \textcite{Few20}, what he calls the \emph{compatibility axiom} is missing.

\subsubsection{Local operations in AQFT}\label{sec:prelim-aqft-operations-locality}
As the central issue of this essay is to discuss the class of physically possible operations, we will need a notion of operation in AQFT, analogous to the (much better-known and -studied) notion on Hilbert space in Definition~\ref{def:hs-operation}. The following definition is motivated by the discussion following that Definition.

\begin{definition}\label{def:aqft-operation}
    Let $\Af$ be a unital *-algebra. An \defn{operation on $\Af$} is a linear map $\G: \Af^*\to\Af^*$ which is \defn{completely positive} (see Section~\ref{sec:prelim-functional-analysis}) and \defn{normalisation-nonincreasing}, meaning that $\G(\om)(\one_\Af) \leq \om(\one_\Af)$ for all positive $\om\in\Af^*$. If it is normalisation-preserving (\ie\ we have equality in the inequality), it is called \defn{non-selective}; otherwise, it is \defn{selective}.
    
    
    
\end{definition}

Operations are not particularly frequently studied in AQFT; however, when they are, their definition (in the case of non-selective operations) is usually given in terms of linear maps $T:\Af\to\Af$ between \textit{observables} rather than maps between states as above: see for example \textcite{CH01, RV10}. Every such map $T$ can be translated into a map $\G:\Af^*\to\Af^*$ by taking its pullback $\widehat T$. One can easily show that in this case, \ie\ if $\G = \widehat T$ and $T$ is completely positive and $T(\one_\Af)\leq\one_{\Af}$ ($T(\one_\Af) = \one_{\Af}$), then $\G$ is completely positive and normalisation-nonincreasing (normalisation-preserving).

However, the converse is generally not true: given a map $\G:\Af^*\to\Af^*$, there is not always a map $T:\Af\to\Af$ which has $\G$ as its pullback. So Definition~\ref{def:aqft-operation} is more general (recall from Section~\ref{sec:prelim-qm-hilbert-space}, item (ii) that the two notions are however equivalent for finite-dimensional Hilbert spaces). We need this more general definition, since an important class of definitions we will discuss are those induced by the state update rule of Fewster-Verch measurements (see Section~\ref{sec:fewster-verch}), which are maps between states and not between observables.

The physical intuition behind an operation in AQFT might not be completely clear---for if a state $\om\in\S(\Af(M))$ is already defined on all of space and time, then what does a `change of state' correspond to?
The answer is that a system having a state $\om\in\S(\Af(M))$ merely means that $\om$ gives the expectation value for a hypothetical measurement of any observable in $\Af(M)$ if the system were left otherwise undisturbed. If the system is indeed left undisturbed, it evolves deterministically through time by the dynamics implicitly assumed by the Diamond Axiom (think: unitary Schrödinger evolution), so that the global state $\om$ is uniquely determined by its behaviour in any neighbourhood of any Cauchy surface. In that sense, a state can be seen to correspond to one `instant of time' (though it is best to interpret also this statement as operational, and not ontological, not least because instants of time differ from observer to observer; the operational meaning of `instant of time' here is, if you will, the time perceived by the operating agent). An operation, on the other hand, corresponds to an intervention upon the system by an outside agent, \eg\ by letting it interact with another system and possibly post-selecting on a measurement outcome (for a selective operation). `Where' and `when' this operation `happens' falls outside of the realm of AQFT proper, but is captured to some extent by the notion of \emph{(extended) causal factorisation}, to be introduced later.

We will also need a notion of locality of operations on AQFTs. Roughly, if an outside agent is local to a region $O$\footnote{Note: this does \emph{not} mean the observer itself is described by the system $\Af(O)$, but merely that it can interact with the system inside $O$.} and intervenes upon the system, resulting in an operation, then that operation should not influence the expectation values of observables in regions causally disjoint from $O$. This leads to the following definition (inspired by \cite{CH01, RV10}); we restrict ourselves to non-selective operations since those will be most relevant to us.
\begin{definition}\label{def:aqft-local-op}
    Let $O\subseteq M$ be a spacetime region. We say that a non-selective operation $\G:\Af(M)^*\to\Af(M)^*$ is \defn{local to $O$} iff for all states $\om\in\S(\Af(M))$, we have
    \begin{equation}
        \G(\om)(a) = \om(a) \text{\qquad for all } a\in\Af(L) \text{ with } L\spacelike O.\footnote{Note that we can't simply say that this should hold `for all $a\in\Af(O^\perp)$', since $O^\perp$ need not be open and therefore need not be a region, so that $\Af$ is not defined on $O^\perp$.}
    \end{equation}
    
    
\end{definition}


It is interesting to note that Haag and Kastler, when originally setting up AQFT, interpreted the algebra $\Af(O)$ as corresponding to operations that can be performed within $O$ rather than to observables: hence, they took `local operations' as more fundamental than `local observables' \cite[p.~2]{HK64}. Their notion of (local) operation is, however, slightly different from ours \cite{RV10}.\footnote{To wit, local operations can be defined via their Kraus representations, meaning that there are local algebra elements $a_i\in\Af(O)$ so that $\om'(b) = \om(\sum_{i=1}^n a^*ba)$. This definition is not entirely equivalent to ours in the general *-algebraic context \cite{RV10}.}

\todo{
The following will be of use to us shortly:
\begin{lemma}\label{lemma:heisenberg-local-op}
    Suppose $\Af(M)$ is separated by its states, meaning that for $a\neq b\in\Af(M)$ there is a state $\om\in\S(\Af(\M))$ such that $\om(a)\neq\om(b)$. If $\G:\Af(M)^*\to\Af(M)^*$ is a non-selective operation and $T:\Af(M)\to\Af(M)$ is a linear map such that $\G = \widehat T$, then $\G$ is local to a region $O\subseteq M$ if and only if $T(a) = a$ for $a\in\Af(L)$ with $L\spacelike O$.
\end{lemma}
\begin{proof}
    Let $a\in\Af(L)$ with $L\spacelike O$. $\G$ is local to $O$ if and only if for any such $a$, it holds that $\om(Ta) = \om(a)$ for all states $\om\in\S(\Af(M))$. Under the assumption that the algebra is separated by its states, the latter is equivalent to $Ta = a$.
\end{proof}
}

\subsubsection{A hybrid model}\label{sec:prelim-aqft-hybrid-model}
Most specific models of AQFT, \ie\ nets of algebras $O\mapsto \Af(O)$ that satisfy the axioms introduced in Section~\ref{sec:prelim-aqft-axioms}, are based on quantum fields: as we have noted before, one possibility is to let the local observable algebra $\Af(O)$ consist of algebraic combinations of fields smeared over test functions that vanish outside $O$ (Eq.~\eqref{eq:smear}). Another relevant example of a specific model is given by \textcite{FV20}, describing their measurement scheme (to be discussed in the next section).

Here we will introduce a simpler model, however, which I call `hybrid' because it incorporates the ideas from Hilbert space quantum mechanics into the formalism of AQFT. This will prove insightful to us in Section~\ref{sec:fv-hybrid-model}, and especially in Section~\ref{sec:bfr-hybrid-model}.

The model consists of a spacetime $M$ satisfying the usual requirements, together with a finite number $n$ of smooth, future-directed inextendible causal curves $\{\g_i : \R\to M\}_{i=1}^n$ and $n$ corresponding Hilbert spaces $\H_i$, which we will assume to be finite-dimensional for simplicity. The curves can be thought of as representing the worldlines of pointlike particles carrying a spin described by $\H_i$, and we will indeed refer to them as \defn{worldlines}. As such, we define the global observable algebra as 
\begin{equation}\label{eq:AfM-is-BH}
    \Af(M) := \B(\H_1\tns\cdots\tns\H_n),
\end{equation}
while for any local region $O\subseteq M$, $\Af(O)$ is the minimal subalgebra of $\Af(M)$ that satisfies
\begin{equation}\label{eq:local-algebra-hybrid}
    \exists t\in\R: \g_i(t)\in O \quad\implies\quad \one_1\tns\cdots\one_{i-1}\tns\B(\H_i)\tns\one_{i+1}\tns\cdots\tns\one_n \subseteq \Af(O) \text{\quad for all } i,
\end{equation}
where $\one_j \equiv \one_{\B(\H_j)}$ is the identity operator on $\H_j$.\footnote{Here $\B(\H_1\tns\cdots\tns\H_n)$ is identified with $\B(\H_1)\tns\cdots\B(\H_n)$ to compactify notation. Also, the $\one_i$'s in Eq.~\eqref{eq:local-algebra-hybrid} should really be read as $\{\one_i\}$.} In words, $O$ contains the observables of the subsystem given by the subset of particles whose worldlines intersect $O$. (Later, in Section~\ref{sec:fv-hybrid-model}, we will generalise the hybrid model slightly by allowing different segments of the same worldline to be associated with different subalgebras of $\Af(M)$ which are not necessarily of the form of Eq.~\eqref{eq:local-algebra-hybrid}.)

It is easily seen that this assignment $O\mapsto \Af(O)$ satisfies Isotony. It also satisfies Einstein Causality, since the algebras of any two causally disjoint regions $O_1 \spacelike O_2$ contain different tensor factors: if $\g_i$ intersects $O_1$ then it does not intersect $O_2$, owing to the fact that it is a causal curve. As for the Diamond Axiom, note that for any region $O\subseteq M$, we already have $\Af(O)\subseteq \Af(D(O))$ by Isotony; and by definition of the domain of dependence $D(O)$, if a worldline $\g_i$ passes through $D(O)$ then it also passes through $O$, so that we have $\Af(D(O)) \subseteq \Af(O)$. Finally, the model satisfies Haag duality, a stronger version of the Haag property stating that for any region $O\subseteq M$, if $a\in\Af(O)'$ then $a$ can be localised in a region contained in $O^\perp$ \cite{Lan96} (here again, a prime denotes the commutator). For let $A\subseteq\{1,\dots,n\}$ be such that $i\in A$ iff $\g_i$ passes through $O$, and let $B := \{1,\dots,n\}\setminus A$; then $\Af(M) = \B(\H_A)\tns\B(\H_B)$ up to reordering of the tensor factors, where 
\begin{equation}\label{eq:hilbert-space-deco}
    \H_A := \bigotimes_{i\in A}\H_i \text{\qquad and \qquad} \H_B := \bigotimes_{i\in B}\H_i.
\end{equation}
Haag duality is then equivalent to the general observation that $(\B(\H_A)\tns\one_B)' = \one_A\tns\B(\H_B)$, which follows from \cite[Theorem~1]{RD75} and the fact that $\B(\H_A)$ has a trivial centre \cite[Cor.\,C.156]{Lan17}. 

Although this `hybrid' model satisfies the basic axioms of AQFT, it is perhaps less physically meaningful than usual models of AQFT, for example because it heavily idealises the particle concept. The model is useful, however, as it enables a direct comparison between Hilbert space QM and AQFT. In particular, it can be seen as an extension of Hilbert space QM, adding, as it were, an explicit relativistic causal structure to the formalism: a state $\om\in\Af(M)$, for example, corresponds precisely to a density operator on the large Hilbert space $\H := \H_1\tns\cdots\tns\H_n$, by the results of Section~\ref{sec:prelim-qm-hilbert-space}. Therefore, the hybrid model can be said to be at least as useful in practical considerations as Hilbert space QM itself.
In any case, the hybrid nature of the model will prove to us in Section~\ref{sec:bfr-hybrid-model}, where it will help us find a connection between two important results in these Hilbert space QM and AQFT \cite{BGNP01, BFR21} and will give us an intuitive picture of why the measurement scheme of \textcite{FV20} shows good causal behaviour.

\todo{
To justify the view that the hybrid model is an extension of Hilbert space QM, it makes sense to also compare their respective notions of locality of operations. To this end, let us find out what it means, given a non-selective operation $\G:\Af(M)^*\to\Af(M)^*$ of the hybrid model, to say that $\G$ is local to a region $O\subseteq M$ in the sense of Definition~\ref{def:aqft-local-op}. This is best discussed in the Heisenberg picture, so let us denote by $T : \B(\H)\to\B(\H)$ the map such that $\G$ is its pullback: $\G = \widehat T$. This map $T$ exists at least if $\H$ is finite-dimensional, as discussed in Section~\ref{sec:prelim-qm-hilbert-space} (where it was denoted as $\L^\dagger$). 
$\B(\H)$ is separated by its states, which is a direct consequence of the Hahn-Banach theorem \cite[Corollary~3.4]{Clu09}; therefore we can apply Lemma~\ref{lemma:heisenberg-local-op}. Concretely, let us again use the notation of the sets $A$ and $B$ as above (Eq.~\eqref{eq:hilbert-space-deco}), such that $\Af(M) = \H_A\tns\H_B$ and $\Af(O) = \B(\H_A)\tns\one_B$. Then Lemma~\ref{lemma:heisenberg-local-op} tells us that $\G$ is local to $O$ if and only if $T(\one_A\tns b) = \one_A\tns b$ for all $b\in\B(\H_B)$. We can now use the following Lemma.

\begin{lemma}
    Let $\H_A$ and $\H_B$ be separable Hilbert spaces and $T:\B(\H_A\tns\H_B) \to \B(\H_A\tns\H_B)$ be positive linear map. Then $T(\one_A\tns b) = \one_A\tns b$ for all $b\in\B(\H_B)$ if and only if there is unital, positive linear map $S:\B(\H_A) \to \B(\H_A)$ such that $T = S\tns\id_B$, \ie\ $T(a\tns b) = S(a)\tns b$.
\end{lemma}
\begin{proof}
    :(
\end{proof}

In conclusion, an operator $T$ is local to $O$ iff it is given by an operation $S$ on the subspace corresponding to the subset of worldlines that intersect $O$. This corresponds to the notion of locality that is used in Hilbert space QM, where Alice and Bob both each have access to a tensor factor of a large Hilbert space $\H := \H_A\tns\H_B$, so that Alice's operations are assumed to be of the form $\L_A\tns\id_B: \D(\H_A\tns\H_B)\to\D(\H_A\tns\H_B)$ (in the Schr\"odinger picture, so $\L_A$ being the Hilbert-Schmidt adjoint of $S$ in the Lemma). We will discuss the latter approach in more detail in Section~\ref{sec:bgnp}.

}

%% file: sections/3-fewsterverch.tex
\section{Fewster and Verch's measurement scheme}\label{sec:fewster-verch}
Despite the long history of AQFT, its operational focus and the fact that it revolves around algebras of observables that ought to be able to be `observed', very little attention has been given to the measurement process itself. Instead, like other approaches to QFT, more attention has gone into exploring the unitary part of quantum evolution. Generalising Von Neumann-like measurement theory to the relativistic context is not straightforward, as some types of measurements can be non-local in one way or another, and an instantaneous state update rule as assumed in non-relativistic measurement theory would not be Lorentz covariant in the relativistic case.

A small but important selection of publications has discussed the measurement process in QFT. Among these are the seminal papers by \textcite{LP31}, \textcite{HK70} and \textcite{AA81} from the previous century, which mainly feature discussions about the issue of Lorentz covariance of state reduction and its implications (whether the state is regarded as ontic or not). More recently, however, general descriptions of the QFT measurement mechanisms themselves have been proposed \cite{OO16, FV20}.\footnote{\label{fn:mmt-problem}Also two other papers of Hellwig and Kraus \cite{HK69,HK70a} should be mentioned here, which were in some sense precursors to Fewster and Verch's \cite{FV20} work, albeit less general. Moreover, it is worth noting that none of the papers mentioned in this paragraph really comes to terms with the quantum \emph{measurement problem}, as explicitly stated by \eg\ \textcite[p.~566 para.~6]{HK70} and \textcite[p.~853 para.~2]{FV20}. See also our discussion on page~\pageref{phily-discussion}.} Among these is the generally covariant approach by \textcite{FV20}, building on ideas from nonrelativistic quantum measurement theory \cite{BLPY16}. It is covered in short form in \textcite{Few20} and I will review it in this section, while making simplifications to their original treatment \cite{FV20, Few20}. This does not come at the cost of generality, as argued in Appendix~\ref{app:semicolons}.

The key idea behind Fewster and Verch's (FV) measurement scheme, which originates from quantum measurement theory \cite{BLPY16}, is to separate the \emph{system} of interest from the \emph{probe}, or measurement apparatus, used to measure properties of the system. Roughly speaking, the probe and system start off uncorrelated at early times but then interact with each other in a compact set of spacetime, so that at late times, the probe can be read off to provide information about an observable of the system. In other words, we adhere to the maxim `prepare early and measure late' \cite{Few20}---where `early' and `late' are relative to the interaction region.

To formalise this, let us introduce two AQFTs defined on the same spacetime $M$, that is, two unital *-algebras $\Af(M)$ and $\Bf(M)$ together with assignments
\begin{equation}
    O\mapsto \Af(O) \qquad\text{and}\qquad O\mapsto \Bf(O)
\end{equation}
of subalgebras to regions of spacetime. These theories respectively describe the system and the probe. The system and probe theories combine into a single large \defn{uncoupled theory} $\Uf \equiv \Af\tns\Bf$ on $M$, where
\begin{equation}
    O\mapsto \Uf(O) := \Af(O) \tns \Bf(O) \quad\text{ for each region } O\subseteq M.
\end{equation}
The interaction between system and probe occurs within a compact set $K\subseteq M$, dubbed the \defn{coupling zone}. To formalise this fact, we assume that the actual experiment is described by a \defn{coupled theory} $\Cf$ which is isomorphic to the uncoupled theory everywhere \emph{outside} the causal hull of $K$. Recall the notation
\begin{equation}
    M^{\pm}_K := M\setminus J^{\mp}(K),
\end{equation}
corresponding to `late times' (upper sign) and `early times' (lower sign) relative to the coupling zone $K$.\footnote{To justify this terminology, note that if $L\subseteq M_K^+$ then there is a causal order, in the sense of Eq.~\eqref{eq:causal-order}, in which $K\leq L$ (and hence there is a Cauchy surface of $M$ with $K$ to its past and $L$ to its future, by Lemma~\ref{lemma:separating-cauchy-surface}). If there is also a causal order in which $L\leq K$, which implies that $L\spacelike K$, then it should not matter whether we call $L$ `early' or `late'; this is indeed the case for the FV measurement scheme, as expressed in \eg\ Propositions~\ref{prop:fv-is-local}, \ref{prop:localisation-of-induced-obs} and \ref{prop:FV-local-2} in Section~\ref{sec:fv-locality} below.}
We describe the equivalence of the uncoupled and coupled theories outside $K$ by assuming the existence of *-isomorphisms
\begin{equation}\label{eq:response-maps-existence}
    \phi_\pm : \Uf(M_K^\pm) \xlongrightarrow{\sim} \Cf(M_K^\pm)
\end{equation}
that are \defn{localisation-preserving}, meaning that for any region $L\subseteq M_K^\pm$, $\phi_\pm\big|_{\Uf(L)}$ maps bijectively onto $\Cf(L)$:
\begin{equation}\label{eq:response-maps-localisation-preserving}
    \phi_\pm\big|_{\Uf(L)} : \Uf(L) \xlongrightarrow{\sim} \Cf(L).
\end{equation}
We also require these maps to agree on $M_K^+\cap M_K^- = K^\perp$:
\begin{equation}\label{eq:response-maps-agree}
    \phi_+\big|_{\Uf(K^\perp)} = \phi_-\big|_{\Uf(K^\perp)}.\footnote{Note again that $K^\perp$ is a region (see footnote~\ref{fn:K-perp-is-region}).}
\end{equation}
$\phi_+$ and $\phi_-$ are referred to as the \defn{advanced} and \defn{retarded} \defn{response maps}, respectively. Note that since $M_K^\pm$ both contain a Cauchy surface for $M$, the Diamond Axiom implies that $\Uf(M_K^\pm) = \Uf(M)$ and similarly for $\Cf$, so that $\phi_\pm$ can be seen as maps between the global observable algebras. As shown in Appendix~\ref{app:semicolons} (Corollary~\ref{cor:simplification} and the preceding discussion), the existence of these response maps requires less assumptions than those proposed by \textcite{FV20}, but they are sufficient to describe the FV measurement process.

We now use the response maps to translate between the uncoupled theory, which is easy to describe, and the coupled theory, which represents (an idealisation of) the actual world. For example, we can say that a state on the coupled theory is `uncorrelated at early times' if under the pullback of the isomorphism $\phi_-^{-1}$, it corresponds to a product state on the uncoupled theory. This is in fact what we assume for the first step of the measurement scheme: since we `prepare early', the initial state is given by
\begin{equation}
    \tilde\om_\s := \widehat{\phi_-^{-1}} (\om\tns\s) \quad\in\S(\Cf(M)),
\end{equation}
where $\om\in\S(\Af(M))$ and $\s\in\S(\Bf(M))$ are system and probe states, respectively.\footnote{Here the hat again denotes pullback (cf.\ item (ii) in Section~\ref{sec:prelim-qm-hilbert-space}): so $\tom_\s(c) = (\om\tns\s)(\phi_-^{-1}(c))$ for all observables $c\in\Cf(M)$.}

Once the interaction has taken place, so that the system and probe have decoupled and we are once again in control of the probe as if it were a factor in an uncoupled theory $\Uf = \Af\tns\Bf$, we measure a probe observable $b\in\Bf(M)$. The fact that this happens at a late time tells us that by performing the procedure to measure (``read off'') $b$ in the uncoupled theory, we are actually measuring the observable
\begin{equation}
    \tilde b := \phi_+(\one_\Af\tns b) \quad\in\Cf(M)
\end{equation}
of the coupled theory. We conclude that the expectation value of the probe observable is given by
\begin{equation}\label{eq:tildes}
    \tom_\s(\tilde b) = (\om\tns\s)(\phi_-^{-1}(\phi_+(\one_\Af\tns b))) = (\om\tns\s)(\Theta(\one_\Af\tns b)),
\end{equation}
where
\begin{equation}\label{eq:scattering-morphism}
    \Theta := \phi_-^{-1} \circ \phi_+ : \Uf(M)\to\Uf(M)
\end{equation}
is called the \defn{scattering morphism}, which is an isomorphism encoding the details of the interaction taking place in $K$. In standard formulations of scattering theory in QFT, the scattering morphism would correspond to the adjoint action of the S-matrix (see \textcite[§II.B]{BFR21} for more discussion of this special case).

Now, our original goal was to measure properties of the \emph{system}. How does measuring the probe observable $b$ help with that? Well, if we can find a system observable, say $a\in\Af(M)$, such that regardless of the state $\om$ of the system, the expectation value of $a$ in state $\om$ matches the expectation value of $\tilde b$ in state $\tom_\s$, then the probe's measurement result allows us to conclude something about the value of $\om(A)$---especially if we repeat the same experiment many times. It is easily verifiable \cite[p.~862]{FV20} that when one defines
\begin{equation}\label{eq:e_s}
    \eta_\s(a\tns b) = \s(b)a \qquad\text{and}\qquad \e_\s(b) = (\eta_\s\circ\Theta)(\one_\Af\tns b)
\end{equation}
for $a\in\Af(M)$ and $b\in\Bf(M)$, then the map $\e_\s$ satisfies
\begin{equation}
    \om(\e_\s(b)) = \tom_\s(\tilde b),
\end{equation}
hence giving an interpretation of a probe measurement in terms a measurement of an \defn{induced system observable} $\e_\s(b)$ which works for any probe observable $b$. (However, as one would expect, $\e_\s$ is in general not surjective, meaning that we are not guaranteed to be able to get information about all system observables by measuring appropriate probe observables: see \eg\ Proposition~\ref{prop:localisation-of-induced-obs}.)

\subsection{State update}
We have now formally described a generally covariant measurement scheme in AQFT on possibly curved spacetimes. One further merit of the FV scheme, however, is that it suggests a rule for the state update resulting from a measurement \cite[§3.3]{FV20}. This is useful when considering the situation where multiple subsequent measurements are performed (Section~\ref{sec:fv-multiple}), and will be of use in our discussion of what operations are physically possible (Section~\ref{sec:bfr}).

Just like the rest of the FV scheme, the phenomenon of state update is treated completely operationally: the updated state simply reflects the expectation values of system observables when one conditions on the outcome of the measurement of a particular probe observable. 
So suppose that a probe observable $b$ has been measured and a result has been noted. We will assume that $b$ is an effect (\ie\ corresponds to a `yes/no' question, see Section~\ref{sec:prelim-functional-analysis}), and that the result is that it has been observed (\ie\ the answer `yes' has been obtained); by the reasoning above, the probability of this result is $\tom_\s(\tilde b)$. Assuming this makes the treatment a lot simpler, since both the measurement and the result are encoded in $b$.%
\footnote{In a Hilbert space setting, for example, this assumption is fully justified since measuring an eigenvalue of an observable $b$ is equivalent to observing the effect which is the projection onto $b$'s corresponding eigenspace---or at least the resulting knowledge updates are the same. Moreover, it corresponds to the principle that an actual measurement always yields only finitely much information, and therefore corresponds to a finite number of `yes/no' questions \cite{BFR21}.}
By an argument similar to the approach in non-relativistic quantum measurement theory and using the basic rules of conditional probabilities, 
it is argued by \textcite[§3.3]{FV20} that the updated (\ie\ conditioned) system state $\om'\in\S(\Af(M))$ is given by
\begin{equation}\label{eq:state-update}
    \om' = \frac{\I_{\s,b}(\om)}{\I_{\s,b}(\om)(\one_\Af)},
\end{equation}
where
\begin{equation}\label{eq:state-update-unnormalised}
    \I_{\s,b}(\om)(a) := (\om\tns\s)(\Theta(a\tns b)) \text{\qquad for } a\in\Af(M)
\end{equation}
where $\s$ is the initial probe state. $\om'$ is also referred to as the \defn{post-selected} system state. Note that the normalisation factor $\I_{\s,b}(\om)(\one_\Af)$ in \eqref{eq:state-update} is precisely $\tom_\s(\tilde b)$, the probability that $b$ is observed. Since the only alternative to $b$ is $\one_\Bf - b$, which occurs with probability $\tom_\s(\widetilde{\one_\Bf-b})$, we see that the \defn{non-selective state update} is
\begin{equation}
    \om'_\ns = \tom_\s(\tilde b) \frac{\I_{\s,b}(\om)}{\I_{\s,b}(\om)(\one_\Af)} + \tom_\s(\widetilde{\one-b}) \frac{\I_{\s,\one-b}(\om)}{\I_{\s,\one-b}(\om)(\one_\Af)}
    = \I_{\s,b}(\om) + \I_{\s,\one-b}(\om) = \I_{\s,\one}(\om)
\end{equation}
where $\one \equiv \one_\Bf$. This state satisfies
\begin{equation}\label{eq:state-update-ns}
    \om'_\ns(a) = \I_{\s,\one}(\om) = (\om\tns\s)(\Theta(a\tns\one_\Bf)) \text{\qquad for } a\in\Af(M),
\end{equation}
so it is simply the partial trace of $\widehat\Theta(\om\tns\s)$.\footnote{Compare this equation to the expectation value of the probe observable in Eq.~\eqref{eq:tildes}; the only subtle difference is that $b$ is intended to be actually measurable, whereas the expectation value of the system observable $a$ discussed here refers to hypothetical measurements. However, a slightly more involved derivation shows that also when a further \emph{FV measurement} is carried out in a region disjoint of $J^-(K)$, the expectation values of measurements of probe observables agree with assuming that the system state is given by Eq.~\eqref{eq:state-update-ns}, \cite[Eq.~(25)]{BFR21}. (This makes use of the assumption of causal factorisation, however, which we have yet to introduce.) So the significance of $\om'_\ns(a)$ is not merely with respect to hypothetical measurements. On the other hand, this argument works only for those system observables which are FV measurable, which is not necessarily the complete algebra.} In particular, it is \emph{independent of the probe observable} $b$ that was measured.

This is contrary to what one is used to in simple treatments of nonrelativistic finite-dimensional Hilbert space quantum measurement when considering measurements directly on the system, instead of considering a probe: in the case of a projective measurement, for example, the non-selective state update (given by the projection postulate as expressed in Lüders' rule \cite{BL09}) depends on the basis in which the system is measured. If one non-selectively measures a probe entangled with the system, however, subsequently tracing out the probe erases the information about the measurement---which is precisely the content of the no-signalling theorem \cite{GRW80} (cf.\ the discussion on page~\ref{einstein-causality}).


\

Before we continue, some philosophical comments are in order. First of all,\label{phily-discussion} one might wonder why it is necessary to introduce the probe theory at all---if one can measure a probe observable, why not directly measure a system observable of interest instead? Those who have indeed been wondering this might well have expected the FV scheme to solve the measurement problem%
\footnote{You are correct to infer from the choice of introduction of this paragraph that your humble author was, initially, one of them\ldots.}
(briefly, the question ``what makes a measurement a measurement?'' \cite{Bru17}, \ie\ why does it induce state reduction?). This ambitious goal is, however, out of reach and not amongst the purposes of the FV scheme \cite[p.~853]{FV20}. The scheme instead gives a formal, operational account of a practical measurement procedure. While the system of interest can be small (\ie\ entirely within the quantum realm), the probe typically has both a microscopic part (which interacts with the system) and a macroscopic part, which contains dials that can be read off by observers. What exactly this `reading off' entails is beyond the scope of the discussion; indeed, the measurement problem is simply postponed from the system to the (macroscopic part of) the probe. This is beneficial, since when the system and probe have decoupled, it gives one the opportunity to do whatever is necessary to `read off' the dials, without unintentionally disturbing or needing further access to the system.

Our second comment is that the formalism is an idealisation: in particular, it might be impossible to fully decouple the system and probe, in which case the identification of the coupled and uncoupled theories is not an exact reflection of reality. What's more, even the coupled theory does not describe the actual, complete world: after all, it does not the describe the macroscopic observers themselves, who in addition to reading off the measurement result might unintentionally disturb the system and probe before, during and after the interaction. We can assume, however, that \emph{for all practical purposes}, the FV scheme describes real-world quantum measurements.

Finally,\label{phily-discussion-update-boundary} the formalism implies no particular view on the question `where' state update happens, be it along a Cauchy surface or the past lightcone of the coupling region, as proposed by \textcite{HK70}. Indeed, the updated state $\om'$ generally differs from $\om$ even in $J^-(K)$, meaning that it prescribes different probabilities to system events that already happened before the interaction with the probe. This does not mean that the act of measuring a probe effect retroactively changes the physics in $J^-(K)$; it simply means that in hindsight, with newly obtained knowledge in the form of the measurement result, hypothetical events in $J^-(K)$ are ascribed a different probability. This ties in to our earlier comments (after Definition~\ref{def:aqft-operation}) and we will return to it at the end of Section~\ref{sec:fv-multiple}.

\

Returning to the formalism, let us introduce some useful terminology and summarise what we have learnt about the FV scheme as we go.
First note that, as we have seen above, an FV measurement on a given system $\Af$ on a spacetime $M$ is uniquely determined by a sextuple $(\Bf, \Cf, K, \phi_\pm, \s, b)$, where $\Bf$ is the probe theory, $\Cf$ is the coupled theory, $K$ is the compact coupling region, $\phi_\pm : \Uf(M_K^\pm)\to\Cf(M_K^\pm)$ are the response maps with $\Uf = \Af\tns\Bf$, $\s\in\S(\Bf(M))$ is the initial state of the probe, and $b\in\Bf(M)$ is the measured probe observable. The first three members of this sextuple determine the scattering morphism $\Theta:\Uf(M)\to\Uf(M)$ resulting from the interaction and providing the interpretation of probe measurements as system measurements, while $\s$ and $b$ determine the state update resulting from such a probe measurement. We call any such sextuple an \defn{FV measurement on $\Af$} if it arises from the above constructions, and we call it \defn{non-selective} if $b=\one_\Bf$.

We find that the \emph{unnormalised} state update $\om \mapsto \I_{\s,b}(\om)$ of an FV measurement is an operation on $\Af$, in the sense of Definition~\ref{def:aqft-operation}. That it is CP follows from the fact that $\I_{\s,b} = \Tr_\Bf \circ \widehat\Theta$, where $\Tr_\Bf$ is partial trace over $\Bf$ and $\widehat\Theta$ is the pullback of the scattering morphism, and the fact that both of these maps are CP, which is easily verifiable.\footnote{\label{fn:pullback-is-operation}For example, $\Theta$ is a *-isomorphism, so it is positive: $\Theta(a^*a) = \Theta(a)^*\Theta(a)$ for $a\in\Af\tns\Bf$. Since $\Theta\tns\id_\Cf$ is also a *-isomorphism for any further unital *-algebra $\Cf$, it is also completely positive; therefore $\Theta$ is CP, which implies that $\widehat\Theta$ is CP, as we noted below Definition~\ref{def:aqft-operation}.}
The fact that $\I_{\s,b}$ is normalisation-nonincreasing follows from the fact that $b$ is an effect (in particular, $b\leq \one_\Bf$). Let us define:
\begin{definition}
    An operation $\G$ on $\Af(M)$ is called \defn{FV-induced in $O$} (or simply \defn{FV-induced}) if there exists an FV measurement $(\Bf, \Cf, K, \phi_\pm, \s, b)$ with coupling zone $K$ contained in $O$ such that $\G = \I_{\s,b}$. We have: if the FV measurement is non-selective ($b=\one_\Bf$), in which case we call $\G$ \defn{non-selectively} FV-induced, then $\G$ is a non-selective (\ie\ normalisation-preserving) operation. Moreover, if $\H$ is a Hilbert space, we call an operation $\L:\D(\H)\to\D(\H)$ FV-induced iff $\widehat{\L^\dagger}$ is FV-induced---recall this notation from Section~\ref{sec:prelim-qm-hilbert-space}.
\end{definition}

In this essay, we are interested in the problem of delineating the class of physically realisable operations. Although the FV scheme was mainly designed to describe measurement processes specifically, it seems to also provide a very general description of how local external agents can perform non-selective operations: the agent simply lets a probe interact locally with the system and subsequently forgets (traces out) the probe, yielding the operation given by Eq.~\eqref{eq:state-update-ns}. Hence, it might well be that the class of physically realisable non-selective operations is precisely formed by the non-selectively FV-induced operations. In any case, we will show in the next subsection and in Section~\ref{sec:bfr} that non-selectively FV-induced operations show good causal behaviour.

\subsection{Locality of FV measurements}\label{sec:fv-locality}
We now outline some locality results which follow directly from the formalism introduced above.

First of all, because the response maps agree on $\Uf(K^\perp)$ according to Eq.~\eqref{eq:response-maps-agree}, we find that the scattering morphism acts trivially on this algebra (cf.\ \cite[Prop.\,3.1(a)]{FV20}): that is,
\begin{equation}\label{eq:FV-local-1}
    \Theta(c) = c \qquad\text{ for any } c\in\Uf(K^\perp).
\end{equation}
By specialising this to the case where $c = a\tns\one_\Bf$ for some $a\in\Af(K^\perp)$, we immediately obtain from the non-selective state update rule~\eqref{eq:state-update-ns}:
\begin{proposition}\label{prop:fv-is-local}
    An operation on $\Af$ which is non-selectively FV-induced in $O$ is local to $O$ in the sense of Definition~\ref{def:aqft-local-op}.
\end{proposition}
Hence, the act of performing a non-selective FV measurement cannot be used to influence results of measurements in regions causally disjoint to $K$, as one would hope. This is generally not true for \emph{selective} probe measurements. However, as Theorem~3.4 of \cite{FV20} shows, the expectation value of an observable $a\in\Af(K^\perp)$ changes under the state update rule only if it is correlated with the measured system observable $\e_\s(b)$. Therefore, this change simply corresponds to an inference by the observer, based on their observing $b$. In particular, a separate observer contained in $K^\perp$ will not know that $b$ has been observed and will not be able to make this inference.
Still, Proposition~\ref{prop:fv-is-local} is in itself not sufficient to completely rule out superluminal signalling in the FV scheme, as we will see in Section~\ref{sec:sorkin}.

Specialising~\eqref{eq:FV-local-1} instead to $c = \one_\Af\tns b$ where $b\in\Bf(K^\perp)$, one obtains from straightforward calculations and the Haag property:
\begin{proposition}[{\cite[Theorem~3.3]{FV20}}]\label{prop:localisation-of-induced-obs}
    Let $b\in\Bf(M)$ be a probe observable and $\e_\s(b)$ the induced system observable. Then for any connected region $L$ containing the coupling zone $K$, $\e_\s(b)$ is localisable in $L$. Furthermore, if $b\in\Bf(K^\perp)$ then $\e_\s(b) = \s(b)\one_\Af$, so no information can be retrieved about the system by measuring $b$.
\end{proposition}
So roughly, one cannot measure system observables without interacting with a region in which they are localised. This reflects part of the heuristic principle that the local observable algebras consist of those observables which can be measured by a process confined to that region.\footnote{It should come as no surprise that Proposition~\ref{prop:localisation-of-induced-obs} reflects that principle, however, since the FV framework was designed with the principle in mind.} Conversely, however, it is an open question whether all observables $a$ in a region $O$ can be measured by an FV measurement with coupling zone contained in $O$ \cite[§VIII]{BFR21}---that is, whether there is an FV measurement $(\Bf, \Cf, K\subseteq O, \phi_\pm, \s, b)$ such that $a = \e_\s(b)$.

\

The following theorem will be of significance to us later in this essay. It can be seen as a weakened version of the statement that $\Theta$ is localisation-preserving (which is generally not true).

\begin{proposition}[{\cite[Proposition~3.1(c)]{FV20}}]\label{prop:FV-local-2}
    If $\Theta$ is the scattering morphism of an FV measurement with coupling zone $K$ and $L^{\pm} \subseteq M_K^{\pm}$ are two regions such that $L^+ \subseteq D(L^-)$, then $\Theta(\Uf(L^+)) \subseteq 
    \Uf(L^-)$.
\end{proposition}
\begin{proof}
    By the Diamond Axiom applied to the coupled theory, we have $\Cf(L^+) \subseteq \Cf(L^-)$. Using the property that the response maps $\phi_\pm$ are localisation-preserving \eqref{eq:response-maps-localisation-preserving} on their respective domains $\Uf(M_K^\pm)$, we therefore get that $\Theta\big|_{\Uf(L^+)}$ is given by the composition
    \begin{equation}
        \Uf(L^+) \xlongrightarrow{\phi_+} \Cf(L^+) \lhook\joinrel\longrightarrow \Cf(L^-) \xlongrightarrow{\phi_-^{-1}} \Uf(L^-),
    \end{equation}
    so that indeed $\Theta(\Uf(L^+)) \subseteq \Uf(L^-)$.
\end{proof}

Note that the results of this section are much more easily obtained in our notation than in Fewster and Verch's original notation \cite[Proposition~A.1]{FV20}.

\subsection{Multiple measurements}\label{sec:fv-multiple}
For our discussion to follow, and to substantiate the general applicability of the FV framework, it will be useful to consider what happens when multiple FV measurements are made on the same system, possibly by multiple observers. \textcite{Few20} and \textcite{FV20} discuss this for the case of two measurements; \textcite{BFR21} generalise it to an arbitrary finite number. We will briefly discuss the results from \cite{BFR21} without focussing too much on the argumentation and proofs.

Consider $n$ distinct precompact regions $\{O_i\}_{i=1}^n$ of $M$, which can be seen as regions containing measuring apparatuses of $n$ distinct observers. An assumption made from the outset in the aforementioned papers and which we will also make here and in Section~\ref{sec:bfr}, is that this set of regions admits a causal order, as defined in Section~\ref{sec:prelim-lorentz}.
Let's denote this causal order by $\leq$ and label the regions such that $O_i \leq O_j \iff i \leq j$. Note that the results that we are about to discuss are independent of the particular choice of causal order, as they should be.

Suppose each party $i$ performs an FV measurement $(\Bf_i, \Cf_i, K_i\subseteq O_i, (\phi_\pm)_i, \s_i, b_i)$, yielding the scattering morphism $\Theta_i: (\phi_-)_i^{-1}\circ (\phi_+)_i$, a *-automorphism of $(\Af\tns\Bf_i)(M)$. We would like to be able to describe the combined results of these $n$ measurements as the result of one `large' measurement in a consistent way, having probe theory $\Bf := \bigotimes_{i=1}^n \Bf_i$ and coupling region $\bigcup_{i=1}^n K_i$. This is encoded in the following important assumption.

\begin{definition}\label{def:causal-factorisation}
    A collection of FV measurements as above satisfies \defn{causal factorisation} if there exist a theory $\Cf$ and response maps $\phi_\pm: \left(\Af\tns\bigotimes_{i=1}^n\Bf_i\right)(M) \to \Cf$ with coupling region $\bigcup_{i=1}^n K_i$ which yield the scattering morphism
    \begin{equation}\label{eq:causal-factorisation}
        \Theta = \ul\Theta_1 \circ \ul\Theta_2\circ\cdots\circ\ul\Theta_n : \left(\Af\tns\bigotimes_{i=1}^n\Bf_i\right)(M) \to \left(\Af\tns\bigotimes_{i=1}^n\Bf_i\right)(M)
    \end{equation}
    Here $\ul\Theta_i$ denotes $\Theta_i$ with identities $\id_{\Bf_j}$ for $j\neq i$ tensored in at the appropriate slots.\footnote{Note that this general, $n$-party definition of causal factorisation follows from the bipartite case, originally introduced in \textcite{FV20}, when making further physically motivated assumptions \cite{BFR21}.}
\end{definition}

This essentially means that $\Theta_1$ `happens first', while $\Theta_n$ `happens last' (note that the scattering morphisms map from late times to early times, which is why the order displayed in Definition~\ref{def:causal-factorisation} appears reversed).
In particular, causal factorisation implies that the global unnormalised state update, given by Eq.~\eqref{eq:state-update-unnormalised}, factors as
\begin{equation}\label{eq:big-state-update}
    \I_{\tns_i\s_i,\tns_i b_i} = \I_{\s_n,b_n}\circ\I_{\s_{n-1},b_{n-1}}\circ\cdots\circ\I_{\s_1,b_1}.
\end{equation}
This is independent of the particular causal order, if multiple orders exist \cite{FV20,BFR21}: in the case of two measurements, for example, if $O_1$ and $O_2$ are causally disjoint, then
\begin{equation}
    \I_{\s_1\tns\s_2,b_1\tns b_2} = \I_{\s_1,b_1}\circ\I_{\s_2,b_2} = \I_{\s_2,b_2}\circ\I_{\s_1,b_1},
\end{equation}
as is easily seen from Eq.~\eqref{eq:state-update-unnormalised} and the locality property of the scattering morphisms \eqref{eq:FV-local-1}.

A natural question to ask is what the expectation value of an observable $b_i\in\Bf_i$ of a given observer $i$ is, given that all observers perform their FV measurements and the initial system state is $\om$.\footnote{Again, the terminology `initial state' is rather confusing when it refers to a state defined on the entire spacetime; it should be seen as the state the system would be in if all parties were to abstain from performing their FV measurements. Cf.\ the discussion on operations after Definition~\ref{def:aqft-operation}.} Let's denote this value by $E_i(b_i; \om)$. By considering the simplified case in which $b_i$ is an effect, one can argue that \cite[Eq.~(21)]{BFR21}
\begin{equation}
    E_i(b_i; \om) = (\om\tns\bigotimes_{i=1}^n\s_i)(\Theta(\one_\Af\tns\ul{b_i})),
\end{equation}
where $\ul{b_i} \in \bigotimes_{i=1}^n\Bf_i(M)$ is $b_i$ tensored with units $\one_{\Bf_j}$ for $j\neq i$.
Using this, the assumption of causal factorisation implies that \cite[Eq.~(33)]{BFR21}
\begin{equation}\label{eq:expected-b-i}
    E_i(b_i;\om) = \left((\I_{\s_{i-1},b_{i-1}}\circ\I_{\s_{i-2},b_{i-2}}\circ\cdots\circ\I_{\s_1,b_1})(\om)\right)(\e_{\s_i}(b_i));
\end{equation}
that is, the expectation value of $b_i$ is what it would have been in the absence of the other observers, if the initial state were given by the state updates of those measurements that happen before $i$ in the causal order. In particular, $E_i(b_i;\om)$ does not depend on the measurements that happen after $i$ in the causal order. Eq.~\eqref{eq:expected-b-i} is again independent of the particular choice of causal order, because Eq.~\eqref{eq:big-state-update} is.

Causal factorisation \eqref{eq:causal-factorisation} is a natural assumption given the physical idea behind performing measurements, and is closely related to the aforementioned maxim of `prepare early and measure late'---compare, for example, Eqs.~\eqref{eq:tildes} and \eqref{eq:expected-b-i}. It can also be argued for as a generalisation of similar results in conventional QFT \cite{BFR21}. In any case, it should be verified in concrete models of system-probe couplings where possible (see \eg\ \textcite[§4]{FV20} and the end of Section~\ref{sec:fv-hybrid-model}).

Finally, for the wavefunction-realist who is interested in the question where state update happens (\ie\ what is the quantum state that `occupies' each given region of spacetime), the results of this section give a partial answer: if Eq.~\eqref{eq:expected-b-i} is extended to give expectation values of hypothetical measurements, it suggests that the updated system state is valid everywhere except in the causal past of the regions where an FV measurement has taken place, and hence that state update takes place `along the past lightcone' of the coupling region. This is also what \textcite{HK70} propose. They remark, however, that one could just as well take the state update to happen along the \emph{future} lightcone. This corresponds to Proposition~\ref{prop:fv-is-local}, saying that the FV state update does not affect observables localisable in the causal complement of the coupling region. In accordance with our comment on page~\pageref{phily-discussion-update-boundary}, however, Eq.~\eqref{eq:expected-b-i} does not require this interpretation: Fewster and Verch themselves say that ``there seems to be no purpose in envisaging a transition from $\om$ to $\om'$ occurring along or near some surface in spacetime'' \cite[p.~867]{FV20}.

\subsection{The hybrid model}\label{sec:fv-hybrid-model}
In this section we will apply the ideas of the FV framework to the hybrid model of AQFT introduced in Section~\ref{sec:prelim-aqft-hybrid-model}. This adds some intuition to the abstract treatment above, but will be especially useful to us later in Section~\ref{sec:bfr-hybrid-model}. Specifically, we will be interested in characterising the set of non-selective operations that are FV-induced.

We assume that both the system and probe theories are given by a hybrid model; for illustration, we consider the case in which both theories involve only one worldline, $\g_A$ and $\g_B$ respectively, associated with Hilbert spaces $\H_A$ and $\H_B$. It follows that the uncoupled theory has global algebra $\Uf(M) = \B(\H_A\tns\H_B)$. The interesting case is where the coupling region $K$ of the FV measurement intersects both worldlines. The question now is what a coupled theory $\Cf$ and its response maps $\phi_\pm$ can look like. Recall that although $\phi_\pm$ are in the first instance defined on the algebras $\Uf(M_K^\pm)$, by the Diamond Axiom they are also *-isomorphisms between the global algebras $\Uf(M)\to\Cf(M)$.

First note that without loss of generality, we can assume that the retarded response map $\phi_-$ is the identity on $\Uf(M)$. We can do this by replacing our coupled theory $\Cf$ by $\Cf'$, defined by 
\begin{equation}
    \Cf'(L) := \phi_-^{-1}(\Cf(L)) \text{ \quad for all regions } L\subseteq M,
\end{equation}
which has response maps $\phi_-' := \id_\Uf$ and $\phi_+' := \phi_-^{-1}\circ \phi_+$. Because $\phi_-$ is a \mbox{*-isomorphism} and is localisation-preserving on all of $M$ as a map from $\Cf \to \Cf'$ by construction, it preserves all the relevant structure to ensure that $\Cf'$ is again an AQFT and that $\phi_\pm'$ satisfy the defining properties of response maps \eqref{eq:response-maps-existence}--\eqref{eq:response-maps-agree}. Moreover, the scattering morphism $\Theta' = \phi_+' = \phi_-^{-1}\circ\phi_+$ stays the same, meaning that measurements and (significantly for us) the state update rule \eqref{eq:state-update-ns} are not affected. So let us drop the primes and assume that $\Cf(M) = \Uf(M) = \B(\H_A\tns\H_B)$ and $\phi_- = \id_\Uf$.
Note that this does not mean that the coupling is trivial, since $\Cf(L)$ can generally differ from $\Uf(L)$ for $L\subseteq M_K^+$.

This non-triviality is manifested in the non-triviality of $\phi_+:\Uf\to\Cf$, which is assumed localisation-preserving on $M_K^+$ \eqref{eq:response-maps-localisation-preserving}; indeed, $\phi_+$ uniquely determines the coupling theory.
This is illustrated in Figure~\ref{fig:smiley-face}. We see that we must slightly generalise our treatment of the hybrid model in Section~\ref{sec:prelim-aqft-hybrid-model} so as to accommodate the fact that in the coupled theory, different segments of the same worldline can be associated with different subalgebras of $\B(\H_A\tns\H_B)$ (however, if a segment of a worldline does not intersect any other of the worldlines, then the associated subalgebra must be constant along the segment, as we will see shortly).

\begin{figure}
    \centering
    \input{smiley-face.tikz}
    \caption{The correspondence between the uncoupled (left) and coupled (right) theories outside of the causal hull of the coupling region, given by the response maps $\phi_\pm$. The labels besides the worldlines mean that if a region $L$ intersects a segment of a worldline, then $\Af(L)$ contains the subalgebra indicated by the label: for example, $\Uf(L) = \one_A\tns\B(\H_B)$ while $\Cf(L) = \phi_+(\one_A\tns\B(\H_B))$. The future time direction in any of the figures is up. (Any resemblance to a real person is purely coincidental.)}
    \label{fig:smiley-face}
\end{figure}

We are now interested in characterising the non-selective FV-induced operations amongst the set of all non-selective operations on $\H_A\tns\H_B$. Recall the non-selective FV state-update rule \eqref{eq:state-update-ns}, in this case
\begin{equation}\label{eq:state-update-hybrid}
    \om'(a) = (\om\tns\s)(\Theta(a\tns\one_B)) = (\om\tns\s)(\phi_+(a\tns\one_B)) \text{\qquad for } a\in\B(\H_A),
\end{equation}
where $\s\in\S(\B(\H_B))$ is a probe state. That is, we take the tensor product with a probe state, apply $\phi_+$, and subsequently trace out the probe.

First of all, note that when not restricting $\phi_+$ any further, Eq.~\eqref{eq:state-update-hybrid} covers all non-selective operations on $\H_A\tns\H_B$. This follows from Stinespring's dilation theorem, which is most widely known in its density-operator formulation \cite{Nielsen-Chuang}, stating that any non-selective operation $\L:\D(\H_A)\to\D(\H_A)$ can be written as
\begin{equation}\label{eq:stinespring}
    \L(\rho) = \Tr_B(u^*(\rho\tns\tau)u) \text{\qquad for all } \rho\in\D(\H_A)
\end{equation}
for some environment system $\H_B$, a density operator $\tau\in\D(\H_B)$ and a unitary \break${u\in\B(\H_A\tns\H_B)}$ depending on $\L$. Using the equivalence of algebraic states and density operators of Proposition~\ref{prop:dh-sbh} and the discussion below it, we find that setting $\phi_+(c) := ucu^*$, which is indeed a *-isomorphism, exactly makes the state-update rule \eqref{eq:state-update-hybrid} into the algebraic equivalent of $\L$.

We therefore have to more closely consider what restrictions are imposed on the response map $\phi_+$, apart from it being a *-isomorphism. This requires looking at the coupled theory inside the causal hull of the coupling region $\ch(K)$. Recall that we assume that the coupled theory, also inside the coupling region, is described by the hybrid model and that $\g_A$ and $\g_B$ are the only worldlines.

\begin{wrapfigure}[16]{l}[1cm]{0pt}
    \centering
    \input{D-O-covering.tikz}
    \caption{}
    \label{fig:D-O-covering}
\end{wrapfigure}

First of all, consider the case in which $\g_A$ and $\g_B$ do not intersect.
Using Lemma~\ref{lemma:covering} in Appendix~\ref{app:geometry}, we can cover $\g_A \cap K$ by a set of future domains of dependence $\{D^+(O_i)\}_{i=1}^n$ which do not intersect $\g_B$, where $O_i\subseteq M$ are regions such that
\begin{equation}\label{eq:caps}
    O_{i+1}\cap\g_A \subseteq \D(O_i)\cap\g_A \text{\qquad for } 1\leq i\leq n-1,
\end{equation}
as illustrated in Figure~\ref{fig:D-O-covering}.
By \eqref{eq:caps}, $L_i := O_{i+1}\cap D(O_i)$ contains all sections of worldlines that $O_{i+1}$ does, so by construction of the hybrid model, $\Cf(L_i) = \Cf(O_{i+1})$. Combining this with the Diamond Axiom, we get
\begin{equation}
    \Cf(O_{i+1}) = \Cf(L_i) \subseteq \Cf(D(O_i)) = \Cf(O_i).
\end{equation}
Applying this $n$ times, we see that the local observable algebra of a small neighbourhood of the future endpoint of $\g_A\cap K$ is a subalgebra of the local algebra of a small neighbourhood of the past endpoint of $\g_A\cap K$. Since the same argument holds for a time-inverted version of $\g_A$, these algebras must actually be equal. In conclusion, if a worldline segment (like $\g_A\cap K$) does not intersect any other worldline segment, the subalgebra associated to that worldline must stay constant along the segment.

Back in the context of the FV framework (Figure~\ref{fig:smiley-face}), this means that in case $\g_A$ and $\g_B$ do not intersect in $\Cf$,
\begin{align}
               \phi_+(\B(\H_A)\tns\one_B) &= \B(\H_A)\tns\one_B \text{\quad and\quad} \phi_+(\one_A\tns\B(\H_B)) = \one_A\tns\B(\H_B), \\
    \text{so \quad} \phi_+ &= \psi_A\tns\psi_B \label{eq:psi+-deco}
\end{align}
where $\psi_A$ and $\psi_B$ are *-automorphisms on $\B(\H_A)$ and $\B(\H_B)$, respectively. That is, there is no interaction between the system $\B(\H_A)\tns\one_B$ and the probe $\one_A\tns\B(\H_B)$, as should be expected from a local theory like AQFT when the pointlike carriers of the two systems do not coincide at any point in spacetime.

In this case, therefore, state update \eqref{eq:state-update-hybrid} amounts to nothing special: due to \eqref{eq:psi+-deco}, $\om'(a) = \om(\psi_A(a))$, which can be seen as a simple unitary evolution of the system itself. Indeed, any *-automorphism $\psi_A$ of $\B(\H_A)$ can be written as $\psi_A(a) = uau^*$ where $u\in\B(\H)$ is unitary, which is a consequence of Wigner's theorem (cf.\ \cite[Prop.\,5.25]{Lan17}).

\

We now turn to the case in which $\g_A$ and $\g_B$ do intersect---though for simplicity we assume that they intersect only once (Figure~\ref{fig:two-intersecting-lines}). A similar argument to the above applies to the segments of the worldlines that do \emph{not} intersect other worldlines, so that the subalgebras are constant along those sections. At the intersection point, however, non-trivial interactions can take place, mixing the subalgebras $\B(\H_A)\tns\one_B$ and $\one_A\tns\B(\H_B)$. In fact we find that any *-automorphism $\phi_+ : \B(\H_A\tns\H_B)\to\B(\H_A\tns\H_B)$ is possible in this case. To see this, define
\begin{equation}
    \g_A^{(-)} := \g_A\big|_{(\infty, u)} \text{\quad and\quad} \g_A^{(+)} := \g_A\big|_{(u, \infty)},
\end{equation}
\begin{wrapfigure}{l}[1cm]{0pt}
    \centering
    \input{two-intersecting-lines.tikz}
    \caption{}
    \label{fig:two-intersecting-lines}
\end{wrapfigure}
where $\g_A(u)$ is the intersection point, and define $\g_B^{(\pm)}$ similarly. We then simply define the coupled theory $\Cf$ by letting $\Cf(L)$ be the minimal subalgebra of $\B(\H_A\tns\H_B)$ containing $\B(\H_A)\tns\one_B$ if it intersects $\g_A^{(-)}$, $\phi^+(\B(\H_A)\tns\one_B)$ if it intersects $\g_A^{(+)}$, $\one_A\tns\B(\H_B)$ if it intersects $\g_B^{(-)}$ and $\phi^+(\one_A\tns\B(\H_B))$ if it intersects $\g_B^{(+)}$.

This theory does not violate the Diamond Axiom, since the subalgebras only change in the intersection point: for example, if $D(O)$ contains the intersection point then $O$ necessarily intersects both $\g_A$ and $\g_B$, as those are assumed to be inextendible (see the discussion before Eq.~\eqref{eq:AfM-is-BH} in Section~\ref{sec:prelim-aqft-hybrid-model} and Figure~\ref{fig:two-intersecting-lines}), so that $\Cf(D(O)) = \B(\H_A\tns\H_B) = \Cf(O)$.
Moreover, this coupling theory leads to the desired locality-preserving scattering morphism $\phi_+: \Uf(M_K^+)\to\Cf(M_K^+)$ of Figure~\ref{fig:smiley-face} and hence to the state update rule~\eqref{eq:state-update-hybrid}.

\

Summarising this section, we have seen that
Stinespring's dilation theorem \eqref{eq:stinespring} implies that Eq.\,\eqref{eq:state-update-hybrid}, with $\phi_+$ ranging over automorphisms of $\B(\H_A)$, covers all non-selective operations on $\D(\H_A)$. 
Furthermore, if we assume that the system theory is an instance of the hybrid model \emph{which has only one worldline} $\g_A$ with associated Hilbert space $\H_A$, then any such automorphism $\phi_+$ can be realised as the response map of an FV measurement: in particular, this can be done by using a hybrid model for the probe theory and letting its worldline $\g_B$ intersect $\g_A$ within the coupling region. Because Eq.\,\eqref{eq:state-update-hybrid} is precisely the nonselective state-update rule for this FV measurement, we conclude that \emph{all non-selective operations $\L$ on $\D(\H_A)$ are FV-induced}. In contrast, any FV measurement in which $\g_A$ intersects no probe worldline within the coupling region can only induce \emph{unitary} operations $\L$ due to \eqref{eq:psi+-deco} (namely, $\L = \psi_A^\dagger$ and $\psi_A(a) = uau^*$ with $u$ unitary).

In Section~\ref{sec:bfr-hybrid-model} we will see that the situation is less straightforward when the system theory contains multiple worldlines.


Finally, note that the hybrid model satisfies causal factorisation (Definition~\ref{def:causal-factorisation}). Roughly, given $n$ probe theories $\Bf_i$ with worldlines $\{\g_{B_i}^j\}_{j=1}^{n_i}$, associated Hilbert spaces $\{\H_{B_i}^j\}_{j=1}^{n_i}$, and coupling regions $K_1 \leq \cdots \leq K_n$, define the coupled theory of the combined measurement as including all worldlines and having global observable algebra $\H_A\tns\bigotimes_{i=1}^n\bigotimes_{j=1}^{n_i} \B(\H_{B_i}^j)$. Then causal factorisation \eqref{eq:causal-factorisation} follows from the fact that for any $1\leq i\leq n-1$, all interactions of $\H_A$ with $\H_{B_i}^j$, since they take place at intersection points of $\g_A$ and $\g_{B_i}^j$ in $\ch(K_i)$, must happen before interactions of $\H_A$ with $\H_{B_{i+1}}^j$, because of the assumption that $K_1 \leq \cdots \leq K_n$ is a causal order.

%% file: figures/smiley-face.tikz
\begin{tikzpicture}
	\begin{pgfonlayer}{nodelayer}
		\node [style=none] (0) at (-2.5, 0) {};
		\node [style=none] (1) at (5.5, 0) {};
		\node [style=none] (2) at (0.25, -1) {};
		\node [style=none] (3) at (2.75, -1) {};
		\node [style=none] (4) at (0.25, 1) {};
		\node [style=none] (5) at (2.75, 1) {};
		\node [style=none] (6) at (0.25, 5) {};
		\node [style=none] (7) at (2.75, 5) {};
		\node [style=none] (8) at (0.25, -5) {};
		\node [style=none] (9) at (2.75, -5) {};
		\node [style=none] (11) at (8, 0) {};
		\node [style=none] (12) at (16, 0) {};
		\node [style=none] (13) at (10.75, -1) {};
		\node [style=none] (14) at (13.25, -1) {};
		\node [style=none] (15) at (10.75, 1) {};
		\node [style=none] (16) at (13.25, 1) {};
		\node [style=none] (17) at (10.75, 5) {};
		\node [style=none] (18) at (13.25, 5) {};
		\node [style=none] (19) at (10.75, -5) {};
		\node [style=none] (20) at (13.25, -5) {};
		\node [style=none] (21) at (5, 4) {};
		\node [style=none] (22) at (8.5, 4) {};
		\node [style=none] (23) at (5, -4) {};
		\node [style=none] (24) at (8.5, -4) {};
		\node [style=none] (25) at (-0.5, 0) {\small $\ch(K)$};
		\node [style=none] (26) at (2, 2.25) {};
		\node [style=none] (27) at (3, 2.25) {};
		\node [style=none] (28) at (1.75, 2.25) {{\footnotesize\color{gray} $L$}};
		\node [style=none] (29) at (12.5, 2.25) {};
		\node [style=none] (30) at (13.5, 2.25) {};
		\node [style=none] (31) at (12.25, 2.25) {{\footnotesize\color{gray} $L$}};
		\node [style=none] (32) at (0.25, -5.5) {\small $\g_A$};
		\node [style=none] (33) at (2.75, -5.5) {\small $\g_B$};
		\node [style=none] (34) at (10.75, -5.5) {\small $\g_A$};
		\node [style=none] (35) at (13.25, -5.5) {\small $\g_B$};
		\node [style=none] (36) at (1.5, -6.5) {$\Uf$};
		\node [style=none] (37) at (12, -6.5) {$\Cf$};
	\end{pgfonlayer}
	\begin{pgfonlayer}{edgelayer}
		\draw [style=gray, bend left=45] (26.center) to (27.center);
		\draw [style=gray, bend right=45] (26.center) to (27.center);
		\draw [bend left] (0.center) to (1.center);
		\draw [bend right] (0.center) to (1.center);
		\draw (4.center) to node [below, sloped] {\klein $\B(\H_A)\tns\one_B$} (6.center);
		\draw (5.center) to node [below, sloped] {\klein $\one_A\tns\B(\H_B)$} (7.center);
		\draw (8.center) to node [below, sloped] {\klein $\B(\H_A)\tns\one_B$} (2.center);
		\draw (9.center) to node [below, sloped] {\klein $\one_A\tns\B(\H_B)$} (3.center);
		\draw [bend left] (11.center) to (12.center);
		\draw [bend right] (11.center) to (12.center);
		\draw (15.center) to node [below, sloped, pos=0.6] {\klein $\phi_+(\B(\H_A)\tns\one_B)$} (17.center);
		\draw (16.center) to node [below, sloped, pos=0.6] {\klein $\phi_+(\one_A\tns\B(\H_B))$} (18.center);
		\draw (19.center) to node [below, sloped] {\klein $\one_A\tns\B(\H_B)$} (13.center);
		\draw (20.center) to node [below, sloped] {\klein $\B(\H_A)\tns\one_B$} (14.center);
		\draw [style=arrow right, bend left=15] (21.center) to node [above] {$\phi_+$} (22.center);
		\draw [style=arrow right, bend right=15] (23.center) to node [below] {$\phi_- = \id_\Uf$} (24.center);
		\draw [style=gray, bend left=45] (29.center) to (30.center);
		\draw [style=gray, bend right=45] (29.center) to (30.center);
	\end{pgfonlayer}
\end{tikzpicture}

%% file: figures/D-O-covering.tikz
\begin{tikzpicture}
	\begin{pgfonlayer}{nodelayer}
		\node [style=none] (0) at (0, -3) {};
		\node [style=none] (1) at (1, 4) {};
		\node [style=none] (2) at (-1.75, -3) {};
		\node [style=none] (3) at (1.75, -3) {};
		\node [style=none] (4) at (0, -1.25) {};
		\node [style=none] (5) at (-1.5, -2) {};
		\node [style=none] (6) at (1, -2) {};
		\node [style=none] (7) at (-0.25, -0.75) {};
		\node [style=none] (8) at (-2, -1) {};
		\node [style=none] (9) at (2.5, -1) {};
		\node [style=none] (10) at (0.25, 1.25) {};
		\node [style=none] (11) at (-2.5, -3) {};
		\node [style=none] (12) at (-1.5, 3.75) {};
		\node [style=none] (13) at (-1.25, 0.25) {};
		\node [style=none] (14) at (3.25, 0.25) {};
		\node [style=none] (15) at (1, 2.5) {};
		\node [style=none] (16) at (-0.25, 2) {};
		\node [style=none] (17) at (2.25, 2) {};
		\node [style=none] (18) at (1, 3.25) {};
		\node [style=none] (19) at (-0.25, 2.75) {};
		\node [style=none] (20) at (2.75, 2.75) {};
		\node [style=none] (21) at (1.25, 4.5) {};
		\node [style=none] (22) at (3.25, -0.5) {\klein\color{gray} $D^+(O_{i})$};
		\node [style=none] (23) at (2, 0.25) {\klein $O_{i+1}$};
		\node [style=none] (24) at (0, -3.75) {\klein $\g_A\cap K$};
		\node [style=none] (25) at (-2.5, -3.75) {\klein $\g_B\cap K$};
	\end{pgfonlayer}
	\begin{pgfonlayer}{edgelayer}
		\draw [in=-75, out=120] (0.center) to (1.center);
		\draw [style=gray, in=150, out=15] (2.center) to (3.center);
		\draw [style=gray, in=-165, out=-30] (2.center) to (3.center);
		\draw [style=gray] (2.center) to (4.center);
		\draw [style=gray] (4.center) to (3.center);
		\draw [style=gray, in=-150, out=-15, looseness=0.75] (5.center) to (6.center);
		\draw [style=gray, bend right=15] (6.center) to (5.center);
		\draw [style=gray] (7.center) to (5.center);
		\draw [style=gray] (7.center) to (6.center);
		\draw [in=-75, out=105] (11.center) to (12.center);
		\draw [style=gray, bend right=15] (16.center) to (17.center);
		\draw [style=gray, bend left=330] (17.center) to (16.center);
		\draw [style=gray] (16.center) to (18.center);
		\draw [style=gray] (18.center) to (17.center);
		\draw [style=gray, in=-180, out=-15] (19.center) to (20.center);
		\draw [style=gray, bend left=15] (19.center) to (20.center);
		\draw [style=gray] (20.center) to (21.center);
		\draw [style=gray] (21.center) to (19.center);
		\draw [style=gray, bend left=15, looseness=0.75] (8.center) to (9.center);
		\draw [style=gray, bend right=15, looseness=0.75] (8.center) to (9.center);
		\draw [style=gray] (8.center) to (10.center);
		\draw [style=gray] (10.center) to (9.center);
		\draw [style=thick, bend left, looseness=0.75] (13.center) to (14.center);
		\draw [style=thick, bend right, looseness=0.75] (13.center) to (14.center);
		\draw [style=thick] (13.center) to (15.center);
		\draw [style=thick] (15.center) to (14.center);
	\end{pgfonlayer}
\end{tikzpicture}

%% file: figures/two-intersecting-lines.tikz
\begin{tikzpicture}[scale=1.2]
	\begin{pgfonlayer}{nodelayer}
		\node [style=none] (0) at (-1.25, -2.5) {};
		\node [style=none] (1) at (-1.25, 2.5) {};
		\node [style=none] (2) at (1.25, 2.5) {};
		\node [style=none] (3) at (1.25, -2.5) {};
		\node [style=none] (4) at (0, 0) {};
		\node [style=none] (5) at (-1, -0.5) {};
		\node [style=none] (6) at (1, -0.5) {};
		\node [style=none] (7) at (0, 0.5) {};
		\node [style=none] (8) at (-0.25, 1.5) {\klein $\g_A^{(+)}$};
		\node [style=none] (9) at (1.5, 1.5) {\klein $\g_B^{(+)}$};
		\node [style=none] (10) at (-0.25, -1.75) {\klein $\g_A^{(-)}$};
		\node [style=none] (11) at (1.5, -1.75) {\klein $\g_B^{(-)}$};
		\node [style=none] (12) at (1, 0.25) {\klein\color{gray} $D(O)$};
		\node [style=none] (13) at (0, -1.5) {};
	\end{pgfonlayer}
	\begin{pgfonlayer}{edgelayer}
		\draw [style=gray] (6.center) to (13.center);
		\draw [style=gray] (13.center) to (5.center);
		\draw [style=gray, bend left=15] (5.center) to (6.center);
		\draw [style=gray, bend right=15] (5.center) to (6.center);
		\draw [style=gray] (6.center) to (7.center);
		\draw [style=gray] (7.center) to (5.center);
		\draw [in=-90, out=90, looseness=0.75] (0.center) to (4.center);
		\draw [in=-90, out=90, looseness=0.75] (4.center) to (2.center);
		\draw [in=-90, out=90, looseness=0.75] (4.center) to (1.center);
		\draw [in=90, out=-90, looseness=0.75] (4.center) to (3.center);
	\end{pgfonlayer}
\end{tikzpicture}

%% file: sections/4-sorkin.tex
\section{Superluminal signalling with quantum operations}\label{sec:sorkin}
In previous sections, we have already encountered several locality aspects that are inherent to AQFT and the FV framework. In particular, we showed that FV-induced operations are local to their coupling regions (Proposition~\ref{prop:fv-is-local}), implying that no superluminal signalling can occur between two spacelike regions when all operations under consideration correspond to FV measurements performed in those regions.

The situation becomes less trivial, however, when one considers three or more operating parties positioned in different regions. This debate was largely sparked by a seminal paper by Rafael Sorkin \cite{Sor93} showing that transferring the ideas of ideal measurements and the state update rule to the relativistic case without restricting the set of allowed operations generally leads to superluminal signalling. I will briefly summarise some of his results here, but will make slight (though unimportant) changes so as to be able to use the language and notation of this essay.

The `naive' state update rule\label{sorkins-rule} Sorkin considers comes down to the same intuition which motivates the assumption of causal factorisation discussed in Section~\ref{sec:fv-multiple} (although it is less rigorously formulated than causal factorisation, since Sorkin considers more general notions of measurement and operation). He considers a finite set of distinct precompact subsets $\{O_i\}_{i=1}^n$ of spacetime (not assumed causally convex) and restricts attention to those cases which admit a causal order $\leq$ in the sense of Eq.~\eqref{eq:causal-order}, where the $O_i$ are labelled such that $O_i\leq O_j \iff i\leq j$. He then proposes that when a measurement or operation is performed in each of the $O_i$, the state updates should be executed in the order given by $\leq$. In particular, the expectation value of a measurement by observer $j$ can only depend on the operations carried out in regions $O_k$ with $k<i$ (noting that the result does not depend on the specific causal order if multiple orders are possible).

The problem arises when one considers three parties Alice, Bob, and Charlie, where $O_A$ and $O_C$ are causally disjoint, while $O_B$ is such that $\bar O_B\subseteq M^+_{\bar O_A}$, $\bar O_B\cap J^+(\bar O_A)\neq\emptyset$, $\bar O_B\subseteq M^-_{\bar O_C}$ and $\bar O_B\cap J^-(\bar O_C)\neq\emptyset$, as visualised in Figure~\ref{fig:sorkin-scenario}. This implies not only that there is a causal order on $\{O_A,O_B,O_C\}$ (cf.\ Eq.~\eqref{eq:causal-order} and the succeeding discussion on page~\pageref{eq:causal-order}) but also that this causal order is unique, or `strict', if you will. I will sometimes refer to such a configuration as a \defn{Sorkin scenario}.
It then turns out that if Alice and Bob apply operations $\L_A$ and $\L_B$ local to their regions $O_A$ and $O_B$, which according to Sorkin's state update rule must be executed in the order $\L_B\circ\L_A$, then Alice's choice of operation $\L_A$ can in general influence expectation values of observables in Charlie's region, thereby establishing superluminal signalling.

This can happen even if Alice's and Bob's operations are non-selective, which we will assume throughout the remainder of this essay.\label{schoolboy} This is motivated by the assumption that Alice and Charlie do not communicate through any other means than possibly Alice's quantum operation (otherwise the question whether that operation establishes signalling would not be very interesting). In particular, Alice cannot communicate any measurement result to Charlie for him to perform post-selection on, meaning her operation will be non-selective to Charlie. As for Bob, although he is partly in Charlie's past and can therefore classically communicate to Charlie, Bob's measurement apparatus generally spans his entire region $O_B$. This means that any measurement result will generally only be available to Bob at a point outside of Charlie's past lightcone (\eg\ at the tip of his future domain of dependence, see Figure~\ref{fig:sorkin-scenario}), so that he cannot communicate the result to Charlie.

The possibility of superluminal signalling might not come as a surprise, as Bob's measurement can be highly non-local: Sorkin even considers the possibility that Bob's region $O_B$ is equal to (a neighbourhood of) an entire Cauchy surface. However, an example presented by Sorkin shows that superluminal signalling can already occur in the simple case of an ideal measurement on a two-qubit system.

\begin{figure}
    \centering
    \input{figures/sorkin.tikz}
    \caption{An example of a scenario as considered in \textcite{Sor93}, capable of superluminal signalling when the rules for observable expectation values and state update are naively generalised to the relativistic setting. Here, as usual, the vertical axis represents time and the horizontal axis represents a space dimension.}
    \label{fig:sorkin-scenario}
\end{figure}
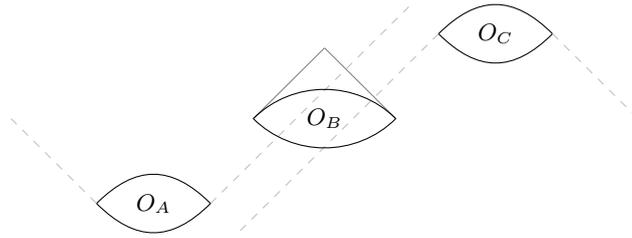

In this example, Alice and Charlie each have access to one qubit system, $\H_A$ and $\H_B$, while Bob has both of these qubits under his control. Suppose the initial system state is the pure state $\k{00}\in\H:=\H_A\tns\H_B$ and Alice can choose whether to leave this state intact or to flip her bit to obtain $\k{10}$. Next, Bob performs a non-selective incomplete Bell measurement; specifically, he measures the observable $\projk{\phi^+}$, where
\begin{equation}
    \kk{\phi^+} = \frac{1}{\sqrt2} \left(\k{00}+\k{11}\right)
\end{equation}
is the usual Bell state. This amounts to applying the operation
\begin{equation}\label{eq:bob-op}
    \rho\quad\longmapsto\quad \projk{\phi^+}\,\rho\,\projk{\phi^+} + \left(\one_\H - \projk{\phi^+}\right) \rho \left(\one_\H - \projk{\phi^+}\right).
\end{equation}
for $\rho\in\D(\H)$. Suppose Alice does change the state to $\k{10}$ before Bob performs his measurement. Since $\k{10}$ is orthogonal to $\kk{\phi^+}$, Bob will obtain result `no' on his $\kk{\phi^+}$-measurement with certainty and the system will be left intact in the state $\k{10}$ (as can easily be verified using Eq.~\eqref{eq:bob-op}). Charlie's subsystem will then be represented by the density operator
\begin{equation}
    \proj{0},
\end{equation}
obtained by tracing out Alice's subsystem. If Alice instead leaves the system in its original state $\k{00}$, Bob's non-selective measurement \eqref{eq:bob-op} collapses the combined system to the state $\frac12 \left(\projk{\phi^+} + \projk{\phi^-}\right)$, so that Charlie's subsystem ends up in the completely mixed state 
\begin{equation}
    \frac12 \left(\proj{0}+\proj{1}\right).
\end{equation}
Thus, if Charlie measures his system in the $\{\k{0},\k{1}\}$-basis and obtains the result $\k{1}$, this unambiguously indicates that Alice did not interfere with the system, thereby establishing superluminal signalling.\footnote{If Charlie measures $\k{0}$, on the other hand, he cannot draw a conclusion about Alice's action; but of course, if they repeat their experiment on a large ensemble of systems, Charlie can be increasingly confident about the message Alice has tried to convey. If they wish to establish superluminal signalling, these experiments should be conducted on a timescale smaller than Alice and Charlie's spatial separation.}

\

One might worry that this example is a too simplified representation of reality to be meaningful---particularly because it is formulated using the language of nonrelativistic Hilbert space quantum theory and assumes the qubits to be located entirely within Alice's and Charlie's respective regions, ignoring the wave-like nature of the electrons or photons carrying encoding qubits.
It would be a gross understatement to say, however, that this language is widely believed to pertain to physical reality to at least some extent---and this example shows that some spacetime realisations of certain Hilbert space operations, \ie\ CPTP maps between spaces of density operators, contradict relativistic causality and hence should be physically impossible.

Superluminal signalling also crops up when considering Sorkin scenarios in QFT, however. This is shown, for instance, in a second example given by \textcite{Sor93} and in an elaborate wave-packet analysis by \textcite{BBBD14}. Hence, Sorkin's paper has sparked a very general programme, which consists in delineating the class of quantum (field) operations that do not violate relativistic causality.

There are many different kinds of approaches to this problem, as we have already discussed in the introduction. An important difference between these approaches is how one formulates the notion of locality and spacelike separation between two parties. While this is hard-wired into the AQFT formalism by the Einstein causality condition (p.~\pageref{einstein-causality}), Hilbert space quantum theory uses what one could call the \emph{tensor product framework} of locality, which assumes that the global system is described by the tensor product of multiple Hilbert spaces, such that each party can only act on their respective Hilbert space. This is what we have done in the above two-qubit example. As I will argue below, deriving general conclusions about the (im)possibility of operations in this framework requires much care. I will illustrate this by briefly reviewing the approach taken by \textcite{BGNP01}.



\subsection{The tensor product framework}\label{sec:bgnp}
What follows is by no means an exhaustive review of \cite{BGNP01}; my purpose here is to only briefly discuss the significance of the paper to the problem of delineating the class of physically possible operations. The paper itself exhibits many other interesting results of both practical and foundational significance.

\textcite{BGNP01} consider a scenario in which two parties, Alice and Charlie, each have access to a quantum system described by finite-dimensional Hilbert spaces $\H_A$ and $\H_C$, respectively. We denote their tensor product by $\H_{AC} := \H_A\tns\H_C$. The objects of study are quantum operations on $\H_{AC}$; for the reasons discussed on page~\pageref{schoolboy}, we restrict our attention to \emph{non-selective} operations, \ie\ linear CPTP maps on $\D(\H_{AC})$, referred to as \emph{superoperators} in \cite{BGNP01}.

First of all, we are interested in whether such operations alone can be used by Alice to signal to Charlie, as formalised in the following definition.
\begin{definition}\label{def:semicausal}
    A non-selective operation $\L_{AC}: \D(\H_{AC}) \to \D(\H_{AC})$ is called \defn{no-signalling from $\H_A$ to $\H_C$} iff for all non-selective operations $\L_A: \D(\H_A)\to\D(\H_A)$ and all initial states $\rho_{AC}\in\D(\H_{AC})$, we have
    \begin{equation}
        \Tr_A \L_{AC}(\rho_{AC}) = \Tr_A \L_{AC}(\L_A\tns\id_C(\rho_{AC})),
    \end{equation}
    where $\id_C:\D(\H_C)\to\D(\H_C)$ is the identity map. It is called \defn{no-signalling} if it is both no-signalling from $\H_A$ to $\H_C$ and from $\H_C$ to $\H_A$, and \defn{signalling} if it is not no-signalling.
\end{definition}

\textcite{BGNP01} call these operations `semicausal' instead of `no-signalling from $\H_A$ to $\H_C$', and `causal' instead of `no-signalling'. I prefer the terminology of Definition~\ref{def:semicausal}, because the definition is asymmetric in $A$ and $C$; moreover, no assumption about the causal order or relative spatiotemporal positions of Alice and Charlie is made at this stage, which use of the term `semicausal' would suggest.
An example of an operation which satisfies none of these definitions (hence enabling signalling in either direction) is the non-selective incomplete Bell measurement considered above. 
Finally, note that \textcite[p.~3]{BGNP01} also consider the use of auxiliary systems by both Alice and Charlie; this makes the definition no more general, however, as shown by their Theorem~1.

Another of the main definitions of \textcite{BGNP01} is that of a localisable operation, \ie\ one which can be implemented using local operations and shared entanglement in the form of auxiliary systems. In this case, we cannot ignore the auxiliary systems.
\begin{definition}\label{def:localisable}
    A non-selective operation $\L_{AC}:\D(\H_{AC})\to\D(\H_{AC})$ is \defn{localisable} iff there are auxiliary systems $\H_R, \H_S$, a shared state $\rho_{RS} \in \H_{RS}$ and non-selective operations $\L_{AR}:\D(\H_{AR})\to\D(\H_{AR})$ and $\L_{CS}:\D(\H_{CS})\to\D(\H_{CS})$ such that\footnote{Here, of course, tensor factors of $\rho_{AC}\tns\rho_{RS}$ should be permuted such that they are acted upon by the respective tensor factors of $\L_{AR}\tns\L_{CS}$.}
    \begin{equation}
        \L_{AC}(\rho_{AC}) = \Tr_{RS}\left( \L_{AR}\tns\L_{CS}(\rho_{AC}\tns\rho_{RS}) \right) \text{\qquad for all } \rho_{AC}\in\D(\H_{AC}).
    \end{equation}
\end{definition}

One of the main purposes of \textcite{BGNP01} is to characterise the sets of no-signalling operations and localisable operations and explore the hierarchy in which they are situated. Let us refer to these sets as `NoSig' and `Loc', respectively (so that they are both subsets of the set of linear CPTP maps on $\B(\H_{AC})$). One can easily argue, for example, that Loc $\subseteq$ NoSig: namely, if some localisable operator were signalling, then Alice and Charlie would be able to signal by using just shared randomness and local operations---which would be possible even if they are spacelike separated, therefore breaking relativistic causality. (An example of a localisable and hence no-signalling operation is complete projective Bell basis measurement \cite[p.\,6]{BGNP01}---which is interesting, because we saw before that Eq.~\eqref{eq:bob-op}), which is an incomplete Bell basis measurement, is signalling.)

On the other hand, not every no-signalling operation is localisable, as shown by example in \textcite[sec.~V]{BGNP01}. One of the main results of \textcite{BGNP01} is, however, that \emph{an operation is semicausal iff it is semilocalisable}---the latter term will be introduced in Definition~\ref{def:semilocalisable} in Section~\ref{sec:bfr-hybrid-model}, where we will discuss this result in more detail.
The authors of \cite{BGNP01} also study stronger notions of localisability, such as those satisfied by operations which can be established by local operations without shared resources or with just shared randomness (instead of shared entanglement). In the hierarchy discussed before, these are of course subsets of Loc.

What is interesting for us is whether and where the class of \emph{physically possible} operations, denoted `Phys', lies in this hierarchy.\footnote{Phys refers to the operations which are somehow physically implementable, by Alice and Charlie and/or any third party, irrespectively of how much shared randomness or entanglement is used. However, this class is meant to be heuristic; as the following discussion suggests, it is debatable whether it can actually be formally defined within the tensor product framework.} Prima facie, one would expect to arrive at a conclusion like
\begin{equation}\label{eq:hierarchy}
    \text{Loc} \ \subseteq\  \text{Phys} \ \subseteq\ \text{NoSig}.
\end{equation}
When one makes the additional assumptions that Alice's and Charlie's Hilbert spaces are both realised by pointlike and spacelike separated physical systems such that Alice and Charlie each have access to exactly one of them, then such a hierarchy might be justified. We will discuss this situation in more detail in Section~\ref{sec:bfr-hybrid-model} using the hybrid model of AQFT.
When generalising away from these idealised assumptions, however, the situation becomes more ambiguous, as we will see shortly. The main worry is that the tensor product Hilbert space framework is only a partial representation of reality---in particular, it is inherently nonrelativistic, so that one has to superimpose a spacetime structure onto the formalism. This can lead to a Hilbert space operation having multiple fundamentally different physical realisations. Let me give some examples.

First of all, it is in general not true that all localisable operations are physically possible. For example, if Alice's system has some spatial extent, as would in fact be true for any truly physical realisation of her quantum system, then her local operation $\L_{AR}$ in Definition~\ref{def:localisable} might well be impossible because it could enable superluminal signalling within Alice's system. This could be realised by letting Alice's system be the tensor product of two subsystems
$\H_A = \H_{A_1}\tns\H_{A_2}$, assumed spacelike separated. Moreover, apart from violation of relativistic causality, there might be further physical phenomena restricting the class of physically possible operations that apply even to the idealised case of pointlike systems (see Section~\ref{sec:conclusion} for a discussion of some of these). Therefore it seems that the first inclusion in \eqref{eq:hierarchy} is not fully justified.

Furthermore, one can argue that some of the operations classified as `signalling' by Definition~\ref{def:semicausal} might actually admit a physical implementation. This is because the physical systems involved might admit some concrete physical (hence relativistic) realisation which has the property that although the operation enables signalling between the systems, there are further physical restrictions which make it impossible to configure \emph{spacelike} parties having access to the respective systems so as to establish \emph{superluminal} signalling.

In the case where both systems $\H_A$ and $\H_C$ are qubits (as in the example of non-selective incomplete Bell measurement), one could think, for example, of realising them as the spins of two electrons which are located very close together, perhaps are even bound to the same atom. In this case, it is not at all straightforward how one could position two spacelike parties in the scenario which each have access to only one of the electrons.\footnote{Of course, one now also has the issue of how Bob is practically able to perform his operation $\L_{AC}$ on two such closely located electrons, but that's another issue altogether.} One would have to take into account the wave-like nature of the \emph{positions} of the electrons themselves;\footnote{One could even consider placing the electrons (or perhaps this would be easier to do with photons) in a superposition of `electron $A$ is at $x$ and electron $B$ is at $y$' and `electron $A$ is at $y$ and electron $B$ is at $x$', where $x\neq y$ are space coordinates.} indeed, the possibility of extracting superluminal signalling would in that case depend on as yet unknown results in the study of indefinite causal orders \cite{Bru14}.\footnote{This programme involves, among other things, studying the effect of spatial superpositions on relativistic spacetime structure and quantum causality. \textcite[Fig.~1]{ZCPB19} provide an example of a situation in which the nature of the relativistic causal relation between two systems profoundly depends on their quantum properties.}

Generally, while we \emph{can} say that signalling between Alice and Charlie in a Sorkin scenario is impossible, Bob's operation is only one of the ingredients that goes into the Sorkin scenario (others being positioning Alice and Charlie in such a way as to have access to exactly one of the subsystems), and so it is not necessarily that operation which renders the scenario impossible.

Another possibility is to consider a qudit with $d=4$, described by a four-dimensional Hilbert space $\H$ and realised in a single, indivisible system with negligible size.\footnote{This can be achieved, for example, by trapping a particle in a potential well such that there is a large gap between the fourth and fifth lowest eigenvalues of the Hamiltonian and restricting attention to the first four eigenstates, which then form the basis vectors of the Hilbert space \cite{LAC16}.} Then one can arbitrarily decompose the Hilbert space into two two-dimensional Hilbert spaces $\H_A$ and $\H_C$ such that $\H = \H_{A}\tns\H_{C}$; but obviously no operation performed on $\H$ can lead to superluminal signalling. Of course, this example is lame since $\H_A$ and $\H_C$ do not correspond to subsystems in the physical sense, but it illustrates once again that whether a given operation in the tensor product framework violates relativistic causality or not depends in part on the particular physical implementation of the involved Hilbert space, regardless of its tensor product factorisation.

In conclusion, we see that in sufficiently general situations, the Hilbert space tensor product framework is somewhat ambiguous on the question what operations are physically realisable, because part of the answer lies outside of this formalism. This is not to say that the framework is not useful: the assumption that the system can be split into two pointlike subsystems each of which can be accessed by either Alice or Charlie is in many cases justifiable. In the next section, we return to the setting of AQFT and the FV scheme. This approach is more complete in the sense that it does not require the aforementioned assumption: instead, the system is described by only one algebra $\Af$ and locality, in the form of Einstein causality, is integrated into the formalism from the outset.

%% file: figures/sorkin.tikz
\begin{tikzpicture}[scale=1.5]
	\begin{pgfonlayer}{nodelayer}
		\node [style=none] (0) at (-2, 0) {};
		\node [style=none] (1) at (0, 0) {};
		\node [style=none] (3) at (-1, 0) {{\small $O_A$}};
		\node [style=none] (5) at (0.75, 1.5) {};
		\node [style=none] (6) at (3.25, 1.5) {};
		\node [style=none] (7) at (2, 1.5) {};
		\node [style=none] (8) at (2, 1.5) {{\small $O_B$}};
		\node [style=none] (9) at (4, 3) {};
		\node [style=none] (10) at (6, 3) {};
		\node [style=none] (12) at (5, 3) {{\small $O_C$}};
		\node [style=none] (13) at (3.5, 3.5) {};
		\node [style=none] (14) at (0.5, -0.5) {};
		\node [style=none] (15) at (7.5, 1.5) {};
		\node [style=none] (17) at (-3.5, 1.5) {};
		\node [style=none] (18) at (2, 2.75) {};
	\end{pgfonlayer}
	\begin{pgfonlayer}{edgelayer}
		\draw [style=dashed edge] (1.center) to (13.center);
		\draw [style=dashed edge] (9.center) to (14.center);
		\draw [style=gray] (5.center) to (18.center);
		\draw [style=gray] (18.center) to (6.center);
		\draw [in=135, out=45, looseness=1.25] (0.center) to (1.center);
		\draw [in=-135, out=-45, looseness=1.25] (0.center) to (1.center);
		\draw [in=135, out=45, looseness=1.25] (9.center) to (10.center);
		\draw [in=-135, out=-45, looseness=1.25] (9.center) to (10.center);
		\draw [style=dashed edge] (17.center) to (0.center);
		\draw [style=dashed edge, in=135, out=-45] (10.center) to (15.center);
		\draw [in=135, out=45] (5.center) to (6.center);
		\draw [in=-135, out=-45] (5.center) to (6.center);
	\end{pgfonlayer}
\end{tikzpicture}

%% file: sections/5-fv-is-causal.tex
\section{FV-induced operations are causally well-behaved}\label{sec:bfr}
In Section~\ref{sec:bfr-generalisation} below we review and slightly generalise a result proven by \textcite[Thm.\,2]{BFR21}, showing that there is no superluminal signalling in a Sorkin scenario when restricting Alice's and Bob's operations to FV-induced operations.
Some intuition behind the proof of this result will be given in Section~\ref{sec:bfr-hybrid-model}, where we study the Sorkin scenario in the hybrid model. In addition, we show that in this model and under suitable conditions, the converse result holds: any operation that is causally well-behaved can be implemented in the FV scheme.

\subsection{Generalisation of a theorem of Bostelmann et al.}\label{sec:bfr-generalisation}
\textcite[Thm.\,2]{BFR21} show that for any observable $c\in\Af(O_C)$, if Alice's and Bob's operations are FV-induced in $O_A$ and $O_B$, respectively, then the expectation value of Charlie's $c$-measurement is independent of Alice's operation.
Because we are (at least I am!) interested in characterising the class of physically possible operations, this result actually does not tell us much: for if it turns out that not all possible operations are FV-induced, then if Bob uses an FV-induced operation, Alice might still be able to superluminally signal by using a non-FV-induced operation. Therefore we will generalise Theorem~2 of \cite{BFR21} slightly by only requiring Alice's operation to be local to $O_A$ (see Definition~\ref{def:aqft-local-op}). The generalisation is simple enough, but interpreting it requires extending the assumption of causal factorisation introduced in Definition~\ref{def:causal-factorisation} to general local operations (as Sorkin did, see page~\ref{sorkins-rule}---but we will try to be slightly more rigorous).

In particular, we adopt the following variant of Eq.~\eqref{eq:expected-b-i} (which, recall, is a result of causal factorisation \eqref{eq:causal-factorisation}). For causally ordered $O_1\leq O_2 \leq \cdots \leq O_n$, if party $i$ measures an observable $c\in\Af(O_i)$ (using the FV approach or by whatever other means, if those turn out to exist) while the other parties ($j\neq i$) perform non-selective operations $\G_j:\S(\Af)\to\S(\Af)$ that are local to $O_j$ (Definition~\ref{def:aqft-local-op}), then the expectation value of $i$'s $c$-measurement, given that the initial system state is $\om$, is
\begin{equation}\label{eq:extended-causal-factorisation}
    E_i(c;\om) = \left(\G_{i-1}\circ\cdots\circ\G_{1}(\om)\right)(c).
\end{equation}
We call this \defn{extended causal factorisation}, and we assume it holds for the class of physically possible operations, for much the same reasons as discussed in Section~\ref{sec:fv-multiple}, pertaining to the operational meaning of operation as being a process that takes an input at early times and has an output at late times (but it should be checked, as far as possible, for suitable candidates of this class). Again, we see that for a region contained in $O_i^\perp$, it does not matter whether we consider it to be `early' or `late' with respect to the operation $\G_i$, since $\G_i$ is local to $O_i$ and hence does not have influence on observables localisable in regions contained in $O_i^\perp$. Accordingly, Eq.~\eqref{eq:extended-causal-factorisation} does not depend on the choice of causal order if multiple orders are possible.

We are now ready to state the following theorem, slightly generalised from \textcite[Thm.~2]{BFR21}.
\begin{theorem}\label{thm:bfr-thm2}
    Let $O_A, O_B,O_C\subseteq M$ be precompact regions which admit a causal order $O_A\leq O_B\leq O_C$, while $\bar O_A\spacelike\bar O_C$. Let $\G_B$ be an operation that is non-selectively FV-induced in $O_B$. Then for any non-selective operation $\G_A$ local to $O_A$ and any initial state $\om\in\S(\Af)$, we have
    \begin{equation}
        \left(\G_B\circ\G_A(\om)\right)(c) = \left(\G_B(\om)\right)(c)  \text{\qquad for all } c\in\Af(O_C).
    \end{equation}
\end{theorem}
The proof of the theorem itself does not use extended causal factorisation. However, only by using extended causal factorisation \eqref{eq:extended-causal-factorisation} can we interpret it as saying that the FV-induced operation $\G_B$ does not enable Alice to signal to Charlie using her local operation $\G_A$.

The proof uses the following geometric lemma, proven in \textcite{BFR21} using a variant of Lemma~\ref{lemma:separating-cauchy-surface}.
\begin{lemma}[{\cite[Lemma~4]{BFR21}}]\label{lemma:geom}
    Let $O_A,O_B,O_C\subseteq M$ be as in Theorem~\ref{thm:bfr-thm2} and let $K\subset O_B$ be compact. Then
    \begin{equation}
        O_C \subseteq D\left(\bar{O_A}^\perp \setminus J^+(K)\right).
    \end{equation}
\end{lemma}

\begin{proof}[Proof of Theorem~\ref{thm:bfr-thm2}]
    Denote the probe theory of Bob's FV-measurement by $\Bf$ and let $c\in\Af(O_C)$. Then $c\tns\one_{\Bf} \in (\Af\tns\Bf)(O_C)$; but we need to find out where $\Theta(c\tns\one_\Bf)$ can be localised, where $\Theta$ is the scattering morphism of Bob's measurement.
    From Lemma~\ref{lemma:geom}, we see that $O_C \subseteq D(\overline{O_A}^\perp \cap M_K^-)$, with $K$ the coupling region of Bob's measurement. Then Proposition~\ref{prop:FV-local-2} applied to $L^+ = O_C$ and $L^- = \overline{O_A}^\perp\cap M_K^-\subseteq M_K^-$ implies that
    \begin{equation}
        \Theta(c\tns\one_\Bf) \in (\Af\tns\Bf)(\overline{O_A}^\perp \cap M_K^-),
    \end{equation}
    so in particular, $\Theta(c\tns\one_\Bf)$ can be localised in a region causally disjoint to $O_A$.\footnote{Note that taking the closure of $O_A$ is necessary because in AQFT one usually considers the assignment of regions to subalgebras to only be defined for open regions, while $O_A^\perp$ is closed \cite[Lemma~A.8]{FV12}.} By the locality of $\G_A$, then,
    \begin{equation}\label{eq:bfr-proof}
        (\G_A(\om)\tns\s)(\Theta(c\tns\one_\Bf)) = (\om\tns\s)(\Theta(c\tns\one_\Bf))
    \end{equation}
    for any initial system state $\om\in\S(\Af)$ and probe state $\s\in\S(\Bf)$. [This follows from the fact that $(\G_A(\om)\tns\s)(a\tns b) = (\G_A(\om))(a)\s(b) = \om(a)\s(b) = (\om\tns\s)(a\tns b)$ for $a\in\Af(L)$ with $L\subseteq O_A^\perp$ and $b\in\Bf(M)$, using the locality of $\G_A$ and noting that $\Theta(c\tns\one_\Bf)$ is a linear combination of such $a\tns b$ terms.] From Eq.~\eqref{eq:state-update-ns} for the non-selective FV state update $\I_{\s,\one_\Bf}$, here denoted by $\G_B$, we see that \eqref{eq:bfr-proof} is synonymous to
    \begin{equation}
        \left(\G_B\circ\G_A(\om)\right)(c) = \left(\G_B(\om)\right)(c),
    \end{equation}
    as required.
\end{proof}


Making further use of extended causal factorisation we can release some of the constraints on $O_B$, leading to perhaps the most succinct formulation of the idea that `all FV-induced operations show good causal behaviour'.

\begin{corollary}\label{cor:fv-is-causal}
    Let $O_B\subseteq M$ be a region and $\G_B$ an operation that is non-selectively FV-induced in $O_B$. Then assuming extended causal factorisation, for any two spacelike regions $O_A$ and $O_C$ such that the triple $\{O_A,O_B,O_C\}$ is causally orderable, $O_B$ does not enable signalling between $O_A$ and $O_C$, \ie\ local operations performed in $O_A$ do not influence the expectation values $E_C(c;\om)$ of observables $c\in\Af(O_C)$ and vice versa, for any initial state $\om\in\S(\Af)$.
\end{corollary}
\begin{proof}
    Denote the causal order by $\leq$. Let's consider the case in which Alice tries to signal to Charlie; so Alice performs $\G_A$ local to $O_A$, Bob performs $\G_B$ and Charlie measures $c\in\Af(O_C)$.
    
    The case of $O_A\leq O_B\leq O_C$ is captured by Theorem~\ref{thm:bfr-thm2}. If $O_C\leq O_A$, then extended causal factorisation~\eqref{eq:extended-causal-factorisation} directly implies that $E_C(c;\om)$ does not depend on $\G_A$. If $O_A\leq O_C\leq O_B$ then \eqref{eq:extended-causal-factorisation} tells us that $E_C(c;\om) = (\G_A(\om))(c) = \om(c)$ because $\G_A$ is local to $O_A$; and finally, if $O_B\leq O_A\leq O_C$ then $E_C(c;\om) = ((\G_B\circ\G_A)(\om))(c) = (\G_B(\om))(c)$, for the same reason. Hence $E_C(c;\om)$ is independent of $\G_A$ and so Alice cannot signal to Charlie; similarly, Charlie cannot signal Alice.
\end{proof}

\subsection{The hybrid model}\label{sec:bfr-hybrid-model}
In Section~\ref{sec:bgnp}, I argued that the Hilbert space tensor product framework of locality is generally not complete enough to determine whether certain operations are physically realisable or not, because the answer can depend on the particular realisation of the Hilbert spaces as relativistic quantum systems. In this section we discuss the problem of superluminal signalling in the hybrid model, which is one way of embedding Hilbert spaces into spacetime (perhaps the most straightforward and simple way, but also a highly idealised one, from a QFT perspective). We will not arrive at a characterisation of the physically realisable operations, as that eventual goal is far beyond the scope of this essay; instead, we will characterise the set of FV-induced operations in the hybrid model (Theorem~\ref{thm:scattering-morphism-hybrid}), which provides an intuitive picture of why these operations, as opposed to general local operations, show good causal behaviour, as we proved for the general case in Theorem~\ref{thm:bfr-thm2}. In addition, we will exploit the hybrid character of this model to make a connection between the FV framework and the purely Hilbert space treatment of \textcite{BGNP01}, arriving via this route at the interesting converse result that in the hybrid model, all operations that show good causal behaviour can in fact be implemented in the FV framework, under a suitable geometrical condition on $O_A$ and $O_C$.

Consider three regions $O_{A,B,C}\subseteq M$ such that there is a causal order $O_A\leq O_B\leq O_C$ and $O_A\spacelike O_C$.\footnote{I am not taking closures of $O_{A,C}$ here, unlike in Theorem~\ref{thm:bfr-thm2}; this is sufficient for the hybrid model.} We will assume Alice's and Charlie's systems to be described by one pointlike quantum system each; so the system theory is given by the global algebra $\Af(M) = \B(\H_A\tns\H_C) =: \B(\H_{AC})$ and two smooth future-directed inextendible causal curves (worldlines) $\g_A$ and $\g_C$, intersecting $O_A$ and $O_C$ respectively and (in the absence of any influence such as a measurement by Bob) associated to the subalgebras $\B(\H_A)\tns\one_C$ and $\one_A\tns\B(\H_C)$. Moreover, we will assume, for simplicity, that Bob's FV measurement is described by a probe theory with only one worldline, $\g_B$, associated to $\H_B$, and a coupled theory which is also an instance of the hybrid model. The following argument will not change too much, however, if Bob's probe consists of multiple pointlike systems (\ie\ multiple worldlines). We will again assume that $\phi_- = \id_\Uf$, so that both the uncoupled and coupled theory have as their global algebra $\B(\H)$ with $\H:= \H_A\tns\H_B\tns\H_C$.

We can carry over many of the arguments from Section~\ref{sec:fv-hybrid-model} to this situation. For example, we see that if the probe $\g_B$ does not intersect $\g_A$ (or $\g_C$), then by the Diamond Axiom of the coupled theory, Bob's scattering morphism $\Theta=\phi_+$ leaves the subalgebra $\B(\H_A)\tns\one_{BC}$ (resp.\ $\one_{AB}\tns\B(\H_C)$) invariant (recall Figure~\ref{fig:D-O-covering}). Here, we say that a *-automorphism $\psi:\B(\H)\to\B(\H)$ \defn{leaves} a *-subalgebra $\mathfrak R\subseteq\B(\H)$ \defn{invariant} iff $\psi(\mathfrak R) = \mathfrak R$.
Therefore, if we want to have any hope of establishing signalling from Alice to Bob, $\g_B$ should intersect both $\g_A$ and $\g_C$; moreover, it should intersect $\g_A$ in Alice's future $J^+(O_A)$ and $\g_C$ in Charlie's past $J^-(O_C)$. (This is not always possible; here we will assume it is, amounting to our `suitable geometrical condition' announced above.)

Importantly, the nature of the Sorkin scenario imposes restrictions on how that can be realised. By definition of $J^+(O_A)$, if $\g_B$ intersects $\g_A$ at, say, $\g_B(t_A)\in J^+(O_A)$ for some $t_A\in\R$, then $\g_B(t) \in J^+(O_A)$ for all $t > t_A$, due to the fact that $\g_B$ is causal. Since $O_A\spacelike O_C$ and thus $J^+(O_A) \cap J^-(O_C) = \emptyset$, we see that $\g_B$ must intersect $\g_C$ \emph{before} it intersects $\g_A$ (\ie\ it intersects $\g_C$ at $\g_B(t_C)$ where $t_C<t_A$ and $\g_B$ is future-directed). This is illustrated in Figure~\ref{fig:the-pretty-one}; again we assume $\g_B$ intersects each of the worldlines only once.

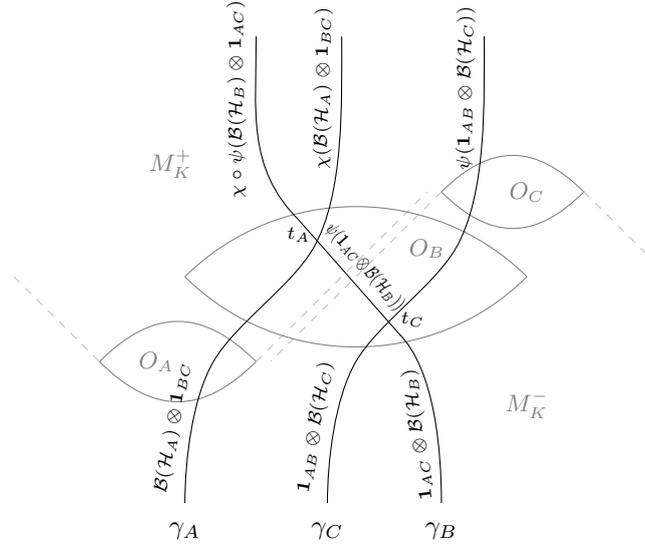
\begin{figure}
    \centering
    \input{the-pretty-one.tikz}
    \caption{Illustration of the coupled theory in the Sorkin scenario of the hybrid model, in the presence of a *-automorphism $\psi$ at $\g_B(t_C)$ and $\chi$ at $\g_B(t_A)$. (See the caption of Figure~\ref{fig:smiley-face} for an explanation of the algebra labels next to the worldlines.) In the \emph{uncoupled} theory, the algebras do not change along the worldlines. In this illustration of the coupled theory, the algebras indicated at late times are precisely those obtained by applying the response map $\phi_+=\Theta$ of Eq.~\eqref{eq:scattering-morphism-hybrid} to the corresponding algebras of the uncoupled theory: see Eq.~\eqref{eq:phi+-uf-oc} for the case of $\g_C$.}
    \label{fig:the-pretty-one}
\end{figure}

Now, what is the set of FV-induced operations? For this we again need to consider what response maps $\phi_+$ (or, if we do not assume $\phi_-=\id_\Uf$, scattering morphisms $\Theta$) are allowed. The subalgebras associated to worldline segments in the coupled theory are shown in Figure~\ref{fig:the-pretty-one}. Recall that due to the Diamond Axiom, these subalgebras are invariant along line segments which do not intersect any other worldline. To the past of $\g_C$, the algebras coincide with the uncoupled theory, for we assume $\phi_- = \id_\Uf$.

Non-trivial interactions now occur at the intersection points $\g_B(t_C)$ and $\g_B(t_A)$. At $\g_B(t_C)$, the algebras associated with worldlines $\g_B$ and $\g_C$ transform by a *-automorphism $\psi:\B(\H)\to\B(\H)$ which, due to the Diamond Axiom, must leave $\one_A\tns\B(\H_{BC})$ invariant; we will assume it can be written as 
\begin{equation}
    \psi = \id_A\tns\psi_{BC}, \text{\quad where \ } \psi_{BC}:\B(\H_{BC})\to\B(\H_{BC}),
\end{equation}
is a *-automorphism, reflecting the fact that $\psi$ is an interaction between the systems at $\g_B$ and $\g_C$.\footnote{\label{fn:assume-identity}The most general form would be $\psi_A\tns\psi_{BC}$. Recall from Section~\ref{sec:fv-hybrid-model} (below Eq.~\eqref{eq:psi+-deco}) that the eventual response map $\phi_+$ may indeed include *-automorphisms of the algebras associated to individual worldlines (in this case $\psi_A$), but that these can be seen as simply corresponding to unitary evolutions of isolated quantum systems. In particular, they do not contribute to signalling. We will take care of these additional automorphisms later, in Eq.~\eqref{eq:phi-with-ups}, since intuitively, it does not make much sense to include $\psi_A$ in an interaction that happens at a point which is not on $\g_A$.}
This implies that the segment of $\g_B$ between $t_C$ and $t_A$ and that of $\g_C$ after $\g_B(t_C)$ have the algebras that are indicated in Figure~\ref{fig:the-pretty-one}.

Similarly, at $\g_B(t_A)$ an interaction $\chi:\B(\H)\to\B(\H)$ occurs between the systems described by the algebras $\B(\H_A)\tns\one_{BC}$ at $\g_A$ and $\psi(\one_{AC}\tns\B(\H_B))$ at $\g_B$ (here the tensor factors are permuted for concision). We therefore assume $\chi$ to be the identity on $\psi(\one_{AB}\tns\B(\H_C))$:
\begin{equation}
    \chi\circ\psi(\one_{AB}\tns c) = \psi(\one_{AB}\tns c) \text{\qquad for all } c\in\B(\H_C)
\end{equation}
This is true if and only if $\psi^{-1}\circ\chi\circ\psi$ is the identity on $\one_{AB}\tns\B(\H_C)$, so we can write
\begin{equation}
    \chi = \psi\circ (\chi_{AB}\tns\id_C)\circ\psi^{-1},  \text{\quad where \quad} \chi_{AB}:\B(\H_{AB})\to\B(\H_{AB})
\end{equation}
is also a *-automorphism.

In addition to these interactions $\psi$ and $\chi$, the response map $\phi_+$ can involve evolutions of the systems when isolated, which can `occur' at any point along a worldline (instead of only at intersections), see footnote~\ref{fn:assume-identity}. Overall, this leads to an extra term $\ups = \ups_A\tns\ups_B\tns\ups_C$ where $\ups_{A,B,C}$ are *-automorphisms of $\B(\H_{A,B,C})$, which should be applied before the interaction terms $\psi$ and $\chi$. We thus find that $\phi_+$ may be expressed as
\begin{equation}\label{eq:phi-with-ups}
    \phi_+ = \chi\circ\psi\circ\ups = (\id_A\tns\psi_{BC})\circ(\chi_{AB}\tns\id_C)\circ(\ups_A\tns\ups_B\tns\ups_C).
\end{equation}
This shows that $\ups_A\tns\ups_B\tns\ups_C$ could just as well be incorporated into $\psi_{BC}$ and $\chi_{AB}$, so we assume that they are identities.

Slightly more involved arguments would show that \eqref{eq:phi-with-ups} also holds when Bob's probe consists of any finite number of systems $\g_{B_i}$, with combined Hilbert space $\H_B := \tns_i\H_{B_i}$. Summarising, we have shown the following.
\begin{theorem}\label{thm:scattering-morphism-hybrid}
    Assume all relevant theories are instances of the hybrid model and suppose a system theory is described by worldlines $\g_A$ with Hilbert space $\H_A$ and $\g_C$ with Hilbert space $\H_C$ which intersect regions $O_A$ and $O_C$, respectively, where $O_A\spacelike O_C$.
    \begin{enumerate}[\normalfont (i)]
        \item If $O_A\leq O_B\leq O_C$ and $O_A\spacelike O_C$, then any FV measurement whose probe theory is described by a Hilbert space $\H_B$ and whose coupling region is contained in $O_B$ has a scattering morphism $\Theta : \B(\H_{ABC})\to\B(\H_{ABC})$ which can be written as
        \begin{equation}\label{eq:scattering-morphism-hybrid}
            \Theta = (\id_A\tns\psi_{BC})\circ(\chi_{AB}\tns\id_C),
        \end{equation}
        where $\psi_{BC}$ and $\chi_{AB}$ are *-automorphisms of $\B(\H_{BC})$ and $\B(\H_{AB})$, respectively.
        
        \item Conversely, any $\Theta$ of the form \eqref{eq:scattering-morphism-hybrid} can be implemented by an FV measurement. In general (namely if $\psi_{BC}$ and $\chi_{AB}$ cannot be decomposed as tensor products of *-automorphisms on the individual spaces $\B(\H_{A,B,C})$), the coupling region $K$ of this FV measurement needs to contain a causal curve that intersects both $\g_C$ and $\g_A$, and therefore might not admit a causal order $O_A\leq K\leq O_C$ (cf.\ footnote~\ref{fn:singularity}).
    \end{enumerate}
\end{theorem}

Note that Eq.~\eqref{eq:scattering-morphism-hybrid} agrees with the labels on the segments of $\g_{A,B,C}$ at late times in Figure~\ref{fig:the-pretty-one} (where we use $\Theta = \phi_-$): for $\g_C$, for example, we have
\begin{align}
    \Cf(O_C) = \phi_+(\Uf(O_C)) &= (\id_A\tns\psi_{BC})\circ(\chi_{AB}\tns\id_C)(\one_{AB}\tns\B(\H_C)) \\
                                &= (\id_A\tns\psi_{BC})(\one_{AB}\tns\B(\H_C)) = \psi(\one_{AB}\tns\B(\H_C)). \label{eq:phi+-uf-oc}
\end{align}
This equation and Figure~\ref{fig:the-pretty-one} show very intuitively why Alice cannot signal to Charlie: Bob's probe only interacts with Alice's system after it has interacted with Charlie's.

Let us see how exactly this fits into the proof of Theorem~\ref{thm:bfr-thm2} (that a \emph{general} FV measurement shows good causal behaviour). Recall that the essential property of FV measurements that makes that proof work is that the scattering morphism satisfies Proposition~\ref{prop:FV-local-2}---in particular, $\Theta(\Uf(O_C)) \subseteq \Uf(L^-)$ whenever $L^-\subseteq M_K^-$ and $O_C\subseteq D(L^-)$. Figure~\ref{fig:the-pretty-one} again shows us intuitively why this is so in the case of the hybrid model: if $O_C\subseteq D(L^-)$ then $D(L^-)$ must contain the intersection point $\g_B(t_C)$, and hence $L^-$ must intersect both $\g_B$ and $\g_C$ by definition of the domain of dependence $D(L^-)$. Therefore $\Uf(L^-)$ contains at least $\one_A\tns\B(\H_{BC})$, which in turn contains $\Theta(\Uf(O_C)) = (\id_A\tns\psi_{BC})(\one_{AB}\tns\B(\H_C))$ (Eq.~\eqref{eq:phi+-uf-oc}).

On the other hand, although a general operation can always be written in the form
\begin{equation}
    \om_{AC}'(a) = (\om_{AC}\tns\s_B)(u (a\tns\one_B) u^*)  \text{\quad for } a\in\B(\H_{AC})
\end{equation}
by Stinespring's dilation theorem, the *-automorphism $a\mapsto uau^*$ of $\B(\H_{AC})$ is not necessarily of the form \eqref{eq:scattering-morphism-hybrid}, and hence does not necessarily satisfy Proposition~\ref{prop:FV-local-2}. Any operation that does enable superluminal signalling, such as the incomplete Bell measurement of Eq.~\eqref{eq:bob-op}, will inevitably need to mix Alice's system into Charlie's, so that $u (\one_{AB}\tns\B(\H_C)) u^* \not\subseteq \one_A\tns\B(\H_{BC})$ and Proposition~\ref{prop:FV-local-2} is not satisfied.

\

Apart from providing intuition as to why FV-induced operations show good causal behaviour, the hybrid model allows us to compare the FV framework with the purely Hilbert space approach of the previously discussed paper by \textcite{BGNP01}. A definition of central importance to that paper, stated in terms of operations on density operators, is the following \cite[Eq.~(19)]{BGNP01}:
\begin{definition}\label{def:semilocalisable}
    A non-selective operation $\L_{AC} : \D(\H_{AC})\to\D(\H_{AC})$ is \defn{semilocalisable} (with communication from $\H_C$ to $\H_A$) iff there is an auxiliary system $\H_B$ and a state $\rho_B\in\D(\H_B)$ such that
    \begin{equation}\label{eq:semilocalisable}
        \L_{AC}(\rho_{AC}) = \Tr_B\left( (\L_{AB}\tns\id_C)\circ(\id_A\tns\L_{BC}) (\rho_{AC}\tns\rho_B) \right) \text{\quad for all } \rho_{AC}\in\D(\H_{AC}).
    \end{equation}
\end{definition}
In words, $\L_{AC}$ is an operation that can be implemented by local operations and quantum communication (of $\H_B$) from Charlie to Alice---cf.\ Definition~\ref{def:localisable} of \emph{localisable} operations.

Without going into much detail, we note that any quantum operation $\L_{AC}:\D(\H_{AC})\to\D(H_{AC})$ is semilocalisable with communication from $\H_C$ to $\H_A$ if and only if it is no-signalling from $\H_A$ to $\H_C$, in the sense of Definition~\ref{def:semicausal}. This was proven in a special case by \textcite{BGNP01} and in full generality by \textcite{ESW02} (in their terminology, an operation is \emph{semicausal} iff it is \emph{semilocalisable}).

Using the equivalences discussed in Section~\ref{sec:prelim-qm-hilbert-space}, Definition~\ref{def:semilocalisable} can be carried over to the algebraic context by taking the Hilbert-Schmidt adjoints of the involved operations. By Theorem~\ref{thm:scattering-morphism-hybrid}, we then find the following. Note that for any precompact and open $O_{A,B,C}$, there is a causal order $O_A\leq O_B\leq O_C$ if and only if $O_B\subseteq {M_{O_A}^+ \cap M_{O_C}^-}$.

\begin{theorem}\label{thm:hybrid}
    Let $O_{A,C}$, $\H_{A,C}$ and $\g_{A,C}$ be as in Theorem~\ref{thm:scattering-morphism-hybrid}; in addition, assume that ${M_{O_A}^+ \cap M_{O_C}^-}$ contains a causal curve that intersects $\g_A$ and $\g_C$.\footnote{\label{fn:singularity}If we assume $\g_{A,C}$ to be inextendible, this seems to be the case in Minkowski spacetime---though I am not completely sure about dimensions higher than 1+1. In a curved spacetime this is not always true, even if $\g_{A,C}$ are inextendible (\eg\ if $M$ is singular, $\g_A$ could be a future-incomplete geodesic \cite[fn.~28]{Lan21} terminating in $O_A$, in which case ${M_{O_A}^+\cap M_{O_C}^-}$ does not even intersect $\g_A$).} 
    Then for any non-selective operation $\L_{AC}:\D(\H_{AC})\to\D(\H_{AC})$, the following are equivalent:
    \begin{enumerate}[{\normalfont(i)}]
        \item $\L_{AC}$ is FV-induced in $M_{O_A}^+ \cap M_{O_C}^-$, where all involved theories are instances of the hybrid model;
        \item $\L_{AC}$ is semilocalisable with communication from $\H_C$ to $\H_A$.
        \item $\L_{AC}$ is no-signalling from $\H_A$ to $\H_C$;
    \end{enumerate}
\end{theorem}
\begin{proof}
    (ii) $\implies$ (i): By using Stinespring's dilation theorem \eqref{eq:stinespring} and expanding the auxiliary system, we see that the definition of semilocalisability does not change if we assume $\L_{AB}$ and $\L_{BC}$ to be given by unitaries ($\rho\mapsto u^*\rho u$) on their respective systems, so that their Hilbert-Schmidt adjoints are also unitarily implementable ($a\mapsto uau^*$), and thus are *-automorphisms of $\B(\H_{AB})$ and $\B(\H_{BC})$, respectively. In the algebraic formulation (see Proposition~\ref{prop:dh-sbh} and succeeding discussion), Eq.~\eqref{eq:semilocalisable} corresponds to
    \begin{equation}\label{eq:semilocalisable-algebraic}
        \widehat{\L_{AC}^\dagger}(\om_{AC})(a) = (\om_{AC}\tns\s_B)\left((\id_A\tns\L_{BC}^\dagger)\circ (\L_{AB}^\dagger\tns\id_C) \left(a\tns\one_B\right)\right)
    \end{equation}
    for some $\s_B\in\S(\B(\H_B))$, all $a\in\B(\H_{AC})$ and all $\om_{AC}\in\S(\B(\H_{AC}))$.
    This is, of course, precisely the state-update rule \eqref{eq:state-update-ns} of a non-selective FV measurement with probe theory $\B(\H_B)$ and a scattering morphism of the form of Eq.~\eqref{eq:scattering-morphism-hybrid}. Hence, by Theorem~\ref{thm:scattering-morphism-hybrid}(ii), any $\L_{AC}$ that is semilocalisable with communication from $\H_C$ to $\H_A$ is FV-induced in ${M_{O_A}^+\cap M_{O_C}^-}$ in the hybrid model.
    
    (i) $\implies$ (ii): Conversely, suppose we have an FV measurement with probe theory $\B(\H_B)$ and coupling region in $M_{O_A}^+ \cap M_{O_C}^-$; then by Theorem~\ref{thm:scattering-morphism-hybrid}(i), its scattering morphism is of the form \eqref{eq:scattering-morphism-hybrid}. This implies that the non-selective state update $\mathcal I_{\s_B,\one_B}$ is of the form \eqref{eq:semilocalisable-algebraic} with $\L_{BC} := \psi_{BC}^\dagger$ and $\L_{AB} := \chi_{AB}^\dagger$, which are operations since $\psi_{BC}$ and $\chi_{AB}$ (from Eq.~\eqref{eq:scattering-morphism-hybrid}) are *-automorphisms. Therefore, $\mathcal I_{\s_B,\one_B}$ is semilocalisable with communication from $\H_C$ to $\H_A$.
    
    (ii) $\iff$ (iii): This follows from the result of \textcite{ESW02}.
\end{proof}

Note that `with communication from $\H_C$ to $\H_A$' refers to the way that Bob implements $\L_{AC}$, namely by letting $\g_B$ first intersect $\g_C$ and then $\g_A$---all within $O_B$. It does not mean that there can be signalling from Charlie to Alice, who are assumed to only be able to control their systems in the regions $O_C$ and $O_A$ (see Figure~\ref{fig:the-pretty-one}).

The equivalence between (i) and (iii) is perhaps the most interesting: it shows that in the case of the hybrid model and under mild assumptions on $O_{A,C}$, the limitations made on the set of quantum operations by restricting to the FV scheme are precisely the limitations imposed by the principle of no superluminal signalling. We leave it as an open question whether this is the case for general models of AQFT.

%% file: figures/the-pretty-one.tikz
\begin{tikzpicture}[scale=1.5]
	\begin{pgfonlayer}{nodelayer}
		\node [style=none] (0) at (-2.75, 0) {};
		\node [style=none] (1) at (0, 0) {};
		\node [style=none] (3) at (-1.75, 0) {\color{gray}\small $O_A$};
		\node [style=none] (5) at (-1.25, 1.5) {};
		\node [style=none] (6) at (4.75, 1.5) {};
		\node [style=none] (8) at (3, 2) {\color{gray}\small $O_B$};
		\node [style=none] (9) at (3.25, 3) {};
		\node [style=none] (10) at (5.75, 3) {};
		\node [style=none] (12) at (4.75, 3) {\color{gray}\small $O_C$};
		\node [style=none] (13) at (3, 3) {};
		\node [style=none] (14) at (0.25, 0) {};
		\node [style=none] (15) at (7, 1.75) {};
		\node [style=none] (17) at (-4.25, 1.5) {};
		\node [style=none] (18) at (-1.25, -2.5) {};
		\node [style=none] (19) at (-0.75, 0) {};
		\node [style=none] (21) at (1.5, 5.75) {};
		\node [style=none] (22) at (1, 2) {};
		\node [style=none] (23) at (1.25, -2.5) {};
		\node [style=none] (24) at (1.75, 0) {};
		\node [style=none] (25) at (4, 5.75) {};
		\node [style=none] (26) at (3.5, 2) {};
		\node [style=none] (27) at (2.5, 0.5) {};
		\node [style=none] (28) at (0.75, 2.5) {};
		\node [style=none] (29) at (0, 5.75) {};
		\node [style=none] (30) at (3.25, -2.5) {};
		\node [style=none] (32) at (3.25, -3) {{$\g_B$}};
		\node [style=none] (33) at (1.25, -3) {{$\g_C$}};
		\node [style=none] (34) at (-1.25, -3) {{$\g_A$}};
		\node [style=none] (35) at (2.75, 0.75) {\klein $t_C$};
		\node [style=none] (36) at (0.75, 2.25) {\klein $t_A$};
		\node [style=none] (40) at (-1.75, -2.25) {};
		\node [style=none] (41) at (0.75, 0.75) {};
		\node [style=none] (42) at (4.75, -0.75) {\color{gray}\small $M_K^-$};
		\node [style=none] (43) at (-1.5, 3.5) {\color{gray}\small $M_K^+$};
	\end{pgfonlayer}
	\begin{pgfonlayer}{edgelayer}
		\draw [style=gray, in=135, out=45, looseness=1.25] (0.center) to (1.center);
		\draw [style=gray, in=-135, out=-45, looseness=1.25] (0.center) to (1.center);
		\draw [style=gray, in=135, out=45, looseness=1.25] (9.center) to (10.center);
		\draw [style=gray, in=-135, out=-45, looseness=1.25] (9.center) to (10.center);
		\draw [style=dashed edge] (1.center) to (13.center);
		\draw [style=dashed edge, in=45, out=-135] (9.center) to (14.center);
		\draw [style=dashed edge] (17.center) to (0.center);
		\draw [style=dashed edge, in=135, out=-45] (10.center) to (15.center);
		\draw [style=gray, in=135, out=45] (5.center) to (6.center);
		\draw [style=gray, in=-135, out=-45] (5.center) to (6.center);
		\draw [in=-120, out=90, looseness=0.75] (18.center) to node [above, sloped] {\klein $\B(\H_A)\tns\one_{BC}$} (19.center);
		\draw [in=60, out=-90, looseness=0.75] (21.center) to node [above, sloped, pos=.25] {\klein $\chi(\B(\H_A)\tns\one_{BC})$} (22.center);
		\draw [in=60, out=-120, looseness=0.75] (22.center) to (19.center);
		\draw [in=-120, out=90, looseness=0.75] (23.center) to node [above, sloped] {\klein $\one_{AB}\tns\B(\H_C)$} (24.center);
		\draw [in=60, out=-90, looseness=0.75] (25.center) to node [above, sloped, pos=.25] {\klein $\psi(\one_{AB}\tns\B(\H_C))$} (26.center);
		\draw [in=60, out=-120, looseness=0.75] (26.center) to (24.center);
		\draw [in=135, out=-90] (29.center) to node [above, sloped, pos=.3, rotate=180] {\klein $\chi\circ\psi(\B(\H_B)\tns\one_{AC})$} (28.center);
		\draw (28.center) to node [above, sloped, pos=.52] {\scalebox{.7}[1.0]{\klein $\psi(\one_{AC}\!\tns\!\B(\H_B))$}} (27.center);
		\draw [in=-45, out=90, looseness=0.75] (30.center) to node [above, sloped, rotate=180, pos=.4] {\klein $\one_{AC}\tns\B(\H_B)$} (27.center);
	\end{pgfonlayer}
\end{tikzpicture}

%% file: sections/6-conclusion.tex
\section{Conclusion}\label{sec:conclusion}
In this essay we have discussed one aspect that limits the class of quantum operations which can be physically realised in principle, namely the constraint of no superluminal signalling arising from the relativistic structure of spacetime. We have mainly worked in the framework of algebraic quantum field theory, focussing on the set of (non-selectively) FV-induced operations as a candidate for this class. We have seen that in general any FV-induced operation at least shows good causal behaviour (Corollary~\ref{cor:fv-is-causal}), and have shown a partial converse to this in the hybrid model (Theorem~\ref{thm:hybrid}).

We have also briefly considered an approach from the Hilbert space formalism, arriving at the conclusion that in general, it remains ambiguous on the question of (im)possibility of operations, because this partly depends on the specific physical realisation of the Hilbert space itself.

The question now arises whether the FV-induced operations precisely form the class of in principle physically realisable operations, when restricting to non-selective operations. At first sight, because of the generality of the FV scheme, it would seem that every physically realisable operation should be FV-induced; but one can never be sure that new physics (possibly beyond AQFT's fixed spacetime background and involving gravity) will not lead to new types of operations \cite{BFR21}---which is why we went to the length of generalising the result of \textcite{BFR21} to Corollary~\ref{cor:fv-is-causal}.

If no superluminal signalling were the only constraint on the possibility of operations, then this Corollary would come close to proving, conversely, that all FV-induced operations are possible, for it says that a single action of performing an FV-induced operation does not enable causally disjoint parties to communicate (no matter what local operations these parties perform). It would remain to investigate situations in which the regions under consideration do not admit a causal order, and more complex scenarios in which apart from Alice, Bob and Charlie, additional parties perform physically realisable operations that are not FV-induced.

Relativistic causality provides, however, not the only limit on quantum operations. Another possibility was introduced by \textcite{Nie97}, who proved (in the Hilbert space setting, but this could easily be carried over to AQFT using \eg\ the hybrid model) that being able to measure certain observables, or to implement certain unitaries, allows one to solve the halting problem. This would contradict the physical Church-Turing thesis, \ie\ the assertion that the functions $\N\to\N$ computable by using physical processes are precisely those that are computable by a Turing machine.\footnote{\textcite{Nie97} calls this the Church-Turing thesis, which, however, often refers to the assertion that the Turing computable functions are precisely those that are computable by \emph{pen-and-paper} methods \cite{sep-church-turing}.} Although this is a heuristic statement which can (in all probability) not be verified with certainty \cite[p.\,3]{Nie97}, its violation would have profound consequences on the foundations of theoretical computer science.
Hence, Nielsen's measurements and unitaries might well be impossible, probably leading to restrictions on unitary system-probe interactions in the FV framework.

Another possible restriction on quantum operations, in particular those corresponding to measurements, is implied by a result of \textcite{Wig52} and \textcite{AY60} also known as the WAY theorem, which identifies limits on the observables one can measure (\eg\ a particle's position) when a non-commuting quantity is subject to a conservation law (\eg\ the combined momentum of the particle and measurement apparatus \cite[p.\,421]{Per02}).

Many other open questions remain. One, which is relevant to the FV scheme regardless of the context of superluminal signalling, is whether all observables of a system can be measured by an FV measurement.
Note that Sorkin's observation that the state update \eqref{eq:bob-op} associated to an ideal measurement of the observable $\projk{\phi^+}$ enables superluminal signalling does not necessarily mean that that observable cannot be measured. A complete projective Bell basis measurement, for example, is no-signalling \cite[p.\,6]{BGNP01} and its measurement result also provides an answer to the question posed by the effect $\projk{\phi^+}$.

It would also be interesting to find out under what assumptions a converse of Corollary~\ref{cor:fv-is-causal} holds, \ie\ that all causally well-behaved operations in AQFT are implementable in the FV framework. By using results from Hilbert space QM \cite{ESW02}, we have proven that this holds in the hybrid model (Theorem~\ref{thm:hybrid}) under a geometrical condition on the region $M_{O_A}^+\cap M_{O_C}^-$. It should not be too difficult to find out in greater detail when this condition is satisfied (see footnote~\ref{fn:singularity}).

Moreover, as we just mentioned, it remains to investigate what happens when the set of regions within which operations are performed does not admit a causal order. Sometimes this means that a region is disconnected or overlaps with another region, which are not very interesting scenarios from a signalling point of view. (The connected components of a causally convex region are necessarily causally disjoint, so it makes most sense to consider one observer to be associated with a connected region; and it is not very interesting to consider signalling between observers that coincide in space and time.) But in more than three spacetime dimensions there are also pairs of connected, disjoint causally convex regions which do not admit a causal order. 
When considering those types of regions, the discussion of multiple measurements (Sections~\ref{sec:fv-multiple} and \ref{sec:bfr} and \cite{BFR21}) would have to be revisited, and new results about locality of the scattering morphism would be necessary, which also include those regions (as opposed to Eq.~\eqref{eq:FV-local-1} and Proposition~\ref{prop:FV-local-2}). For this purpose, one would probably need to extend our definition of the response maps \eqref{eq:response-maps-existence}--\eqref{eq:response-maps-agree} (arriving at Fewster and Verch's original definition, which in our notation is equivalent to Eq.~\eqref{eq:restriction-property}).

Finally, as we have mentioned on several occasions along the way (\eg\ p.\,\pageref{phily-discussion}), the current treatment does not arrive at a solution to the measurement problem. It might well be that we can only come to a truly satisfactory account of the class of physically realisable operations once we have a theory that incorporates observers as parts of the physical system. In the meantime, there is only one way to be certain about the possibility of operations\ldots\ experiment!

%% file: sections/app-isomorphisms.tex
\section{Simplification of Fewster and Verch's notation}\label{app:semicolons}
The exposition of the FV measurement scheme presented in \textcite{FV20} and \textcite{Few20} uses a slightly different, yet more complicated notation than ours, as introduced in Section~\ref{sec:prelim-aqft}. In this section I argue why I think that at least for the present purposes, their notation is not more general, and I prove that when adopting our notation, our set of assumptions on the response maps \eqref{eq:response-maps-existence}--\eqref{eq:response-maps-agree} is as least as general as Fewster and Verch's (in my eyes more complicated) corresponding definition \cite[Eq.~(3.2)]{FV20}.

Fewster and Verch's notation arises when one takes a categorical approach to AQFT (as introduced in \eg\ \textcite{FV15}), where the assignment $O\mapsto \Af(O)$ is expressed by letting $\Af$ be a functor between (roughly) the category of globally hyperbolic spacetimes and the category of unital *-algebras. This means that $O$ is viewed not just as a region in $M$, but as a globally hyperbolic spacetime \emph{an sich}, with metric and time orientation inherited from $M$ (note that when $M$ is globally hyperbolic and $O\subseteq M$ is a region, \ie\ an open and causally convex subset of $M$, then $O$ itself is globally hyperbolic \cite{FV15}).
Because in this case one considers the general category of unital *-algebras instead of restricting attention to subalgebras of one fixed global algebra $\Af(M)$, it is necessary to distinguish between $\Af(O)$ (the algebra assigned to $O$ as a spacetime \textit{an sich}) and the actual subalgebra of $\Af(M)$ describing the observables of $M$ which can be measured by an observer confined to $O$. The latter is notated as $\Af(M;O)\subseteq\Af(M)$. The algebras are assumed equivalent, however, through the existence of *-isomorphisms $\a_{M;O} : \Af(O) \to \Af(M;O)$ for each causally convex region $O$ of any globally hyperbolic spacetime $M$, referred to as \defn{compatibility maps}. These maps are assumed to be consistent amongst each other in the sense that whenever $O_1\supseteq O_2\supseteq O_3$, they obey the
\begin{flalign}\label{eq:compy-axiom}
    \text{\textbf{compatibility axiom}:} && \a_{O_1;O_2} \circ \a_{O_2;O_3} = \a_{O_1;O_3}. &&
\end{flalign}

One of the merits of this notation is that it captures the idea that for any causally convex region $O$ of a globally hyperbolic spacetime $M$, the algebra $\Af(O)$ is completely independent of the world outside $O$, because $\Af(O)$ does not have to be a subalgebra of $\Af(M)$. Similarly, the spacetime $M$ might actually be part of a larger spacetime $N$, such that $\Af(M)$ is isomorphic to a subalgebra of $\Af(N)$.

However, assuming all AQFT considerations to be done within a fixed (though still general!) `global' spacetime $M$ does not seem to be a very limiting assumption. If we do that, we can simply redefine the algebras $\Af(O)$ to be subalgebras of $\Af(M)$ by composing the functor $\Af$ with the compatibility maps. Since the compatibility maps are isomorphisms, this does not change the physics in the local regions.
In this way, therefore, one arrives back at the notation as we introduced it in Section~\ref{sec:prelim-aqft}, and which is widely used in other places in literature \cite{FR20, Hal07}.

\

Let us now find out what the definitions that go into the FV measurement scheme as originally set out in \cite{FV20}, in particular those corresponding to the equivalence between the coupled and uncoupled theories outside the coupling region $K$, translate to in our notation.
In Fewster and Verch's notation, let us write $\u_{L,L'}$ for the compatibility maps of the uncoupled theory $\Uf$ and $\g_{L,L'}$ for those of the coupled theory $\Cf$. Instead of our definitions \eqref{eq:response-maps-existence}--\eqref{eq:response-maps-agree}, Fewster and Verch \cite[§3.1]{FV20} assume the existence of isomorphisms
\begin{equation}
    \chi_L : \Uf(L) \to \Cf(L)
\end{equation}
indexed by regions $L \subseteq M\setminus\ch(K)$, such that for any two such regions $L'\subseteq L$, $\chi_L$ and $\chi_{L'}$ agree up to compatibility maps: \ie, the diagram 
\tikzset{
    verticaal/.style={anchor=south, rotate=90, inner sep=.5mm}
} 
\begin{equation}\label{eq:chi-condition}
\begin{tikzcd}
    \Uf(L') \arrow{r}[name=u]{\u_{L;L'}} \dar["\chi_{L'}","\sim" verticaal]  & \Uf(L) \dar["\chi_L","\sim" verticaal] \\
    \Cf(L') \arrow{r}[name=g]{\g_{L;L'}}                       & \Cf(L)
    \arrow[to path={(g) node[midway, scale=1.2]{$\circlearrowright$} (u)}]
\end{tikzcd}
\end{equation}
commutes \cite[Eq.~(3.2)]{FV20}.

We have the following.

\begin{proposition}\label{prop:semicolons-1}
    There exist isomorphisms $\chi_L:\Uf(L)\to\Cf(L)$ satisfying \eqref{eq:chi-condition} if and only if there exist isomorphisms $\phi_L:\Uf(M;L)\to\Cf(M;L)$ satisfying the restriction property 
    \begin{equation}\label{eq:restriction-property}
        \phi_{L'} = \phi_L\big|_{\Uf(L')}.
    \end{equation}
    Here, $L$ ranges over regions contained in $M\setminus\ch(K)$.
\end{proposition}
\begin{proof}
Assume that the $\chi_L$ isomorphisms exist. For each region $L\subseteq M\setminus\ch(K)$, define $\phi_L$ by composing $\chi_L$ with the appropriate compatibility maps: viz.\ $\phi_L$ is the unique morphism such that the diagram
\begin{equation}\label{eq:chi-phi-diagram}
\begin{tikzcd}
    \Uf(L) \arrow{r}[name=u]{\u_{M;L}}[below]{\sim} \dar["\chi_L","\sim" verticaal]  & \Uf(M;L) \dar[dashed,"\exists!\,\phi_L", "\sim" verticaal] \\
    \Cf(L) \arrow{r}[name=g]{\g_{M;L}}[below]{\sim}      & \Cf(M;L)
    \arrow[to path={(g) node[midway, scale=1.2]{$\circlearrowright$} (u)}]
\end{tikzcd}
\end{equation}
commutes.

Now let $L'\subseteq L$ be two such regions. Then proving the restriction property comes down to proving that the black subdiagram of
\begin{equation}
\begin{tikzcd}[gray]
\black{\Uf(M;L)}    \dar[black,"\phi_L","\sim" verticaal]    & \lar[""{name=11},"\u_{M;L}"',"\sim"]     \Uf(L)    \dar["\chi_L","\sim" verticaal]      & \lar[""{name=12},hook',"\u_{L;L'}"']     \Uf(L')    \dar["\chi_L'","\sim" verticaal]   \rar[""{name=13},"\u_{M;L'}","\sim"'] &   \black{\Uf(M;L')}   \dar[black,"\phi_L'","\sim" verticaal]    \ar[lll,bend right,hook',black]  \\
\black{\Cf(M;L)}    & \lar[""{name=21},"\g_{M;L}"',"\sim"]     \Cf(L)     & \lar[""{name=22},hook',"\g_{L;L'}"']     \Cf(L')     \rar[""{name=23},"\g_{M;L'}","\sim"'] &   \black{\Cf(M;L')}   \ar[lll,bend left,hook',black]
\arrow[draw=none,to path={(21)+(0,.7) node{$B$}}]
\arrow[draw=none,to path={(22)+(0,.7) node{$C$}}]
\arrow[draw=none,to path={(23)+(0,.53) node{$D$}}]
\arrow[draw=none,to path={(12)+(0,.7) node{$A$}}]
\arrow[draw=none,to path={(22)+(0,-.5) node{$E$}}]
\end{tikzcd}
\end{equation}
commutes, where the unlabelled hooked arrows denote inclusion. The black subdiagram does indeed commute, since the five lettered subdiagrams commute: $A$ and $E$ do because of the compatibility axiom \eqref{eq:compy-axiom}, $B$ and $D$ because of definition \eqref{eq:chi-phi-diagram}, and $C$ because of assumption \eqref{eq:chi-condition}.

Conversely, if one assumes the existence of isomorphisms $\phi_L$ satisfying~\eqref{eq:restriction-property}, one can define $\chi_L$ by (again) requiring \eqref{eq:chi-phi-diagram} to commute. Then a similar argument to the above shows that \eqref{eq:chi-condition} is satisfied.
\end{proof}

This proposition translates Fewster and Verch's definition into our notation. At least for the purposes of this essay, however, it is sufficient to consider the isomorphisms $\phi_L$ where $L$ is contained in $M_K^+$, in $M_K^-$, or in both (which is indeed the case whenever one assumes the existence of a causal order on $L$ and the coupling region $K$). By the restriction property \eqref{eq:restriction-property}, all such $\phi_L$ are determined by the two maps $\phi_\pm := \phi_{M_K^\pm}$. One can check that Fewster and Verch's definition of the advanced and retarded response maps \cite[Eq.\,(3.4)]{FV20} simply translates to $\phi_{\pm}$ in our notation. Therefore, these maps are also the only ones needed to define the scattering morphism \eqref{eq:scattering-morphism}. This, together with the following result, motivates our definitions \eqref{eq:response-maps-existence}--\eqref{eq:response-maps-agree}.

\begin{corollary}\label{cor:simplification}
    If there exist isomorphisms $\chi_L:\Uf(L)\to\Cf(L)$ indexed by regions $L\subseteq M\setminus \ch(K)$ and satisfying \eqref{eq:chi-condition}, then there exist isomorphisms $\phi_\pm: \Uf(M;M_K^\pm)\to\Cf(M;M_K^\pm)$ that are localisation-preserving~\eqref{eq:response-maps-localisation-preserving} and which agree on $\Uf(M;M_K^+\cap M_K^-)$ \eqref{eq:response-maps-agree}.
\end{corollary}
\begin{proof}
    Given $\chi_L$, define $\phi_L$ as in Proposition~\ref{prop:semicolons-1} and set $\phi_\pm := \phi_{M_K^\pm}$. Then both required properties of $\phi_\pm$ follow directly from the restriction property~\eqref{eq:restriction-property}.
\end{proof}

Finally, note that the notation of the FV scheme can probably be simplified even further by redefining the coupled theory and the response maps so that we can assume that the retarded response map $\phi_-$ is the identity---as we have in fact done for our hybrid model (Section~\ref{sec:fv-hybrid-model})---so that the coupled and uncoupled theories are defined on the same global algebra and the scattering morphism $\Theta$ is simply $\phi_+$. I have chosen not to do this in the general case because it supports the physical argumentation behind the construction of the scattering morphism in the beginning of Section~\ref{sec:fewster-verch}.

%% file: sections/app-geometry.tex
\section{Geometrical lemmas}\label{app:geometry}
\begin{lemma}[Cf.\ {\cite[Lemma~3]{BFR21}}]\label{lemma:separating-cauchy-surface}
    Let $K$ and $L$ be compact subsets of a globally hyperbolic spacetime $M$ such that $L \subseteq M^+_{K}$. Then there is a Cauchy surface $\Sigma$ of $M$ such that $\Sigma\subseteq M^+_K \cap M^-_L$.
\end{lemma}
\begin{proof}
    Since $M_K^+$ is a causally convex region of $M$ by Lemma~\ref{lemma:ccompl-is-cconvex}, it is globally hyperbolic. By a result of \textcite{Ger70} (see \cite[Prop.~4]{BS03}), then, there exists a continuous, surjective function $t:M_K^+\to \R$ which has Cauchy surfaces of $M_K^+$ as its level sets and which is strictly increasing on all future-directed causal curves. Since $L$ is compact and $t$ continuous, $t(L)$ has a minimal element $t_L$. Take any $t' < t_L$ and set $\Sigma$ equal to the level surface $t^{-1}(\{t'\})$. Then $\Sigma\subseteq M_L^-$ and $\Sigma$ is a Cauchy surface for $M_K^+$. Then it is also a Cauchy surface for $M$: this follows from the observation that every inextendible timelike curve $\g:\R\to M$ restricts to an inextendible timelike curve $\g'$ in $M_K^+$. For any future endpoint of $\g'$ in $M_K^+$ is also a future endpoint for $\g$ in $M$, and similarly for past-endpoints if $\g$ does not intersect $J^-(K)$. If $\g$ does intersect $J^-(K)$, then, being timelike, it crosses the boundary of $J^-(K)$ exactly once, say at $\g(u)$, where $u\in(a,b)$. Noting that $M_K^+$ is closed because $K$ is compact \cite{FV12}, we have $\g' = \g\big|_{(u,b)}$. Now we see that $\g'(t) \to \g(u) \notin M_K^+$ as $t\to u$, so $\g'$ has no past endpoint. In conclusion, $\g'$ is inextendible in $M_K^+$, meaning that it, and hence $\g$, intersects $\Sigma$ in exactly one point.
\end{proof}

\begin{lemma}\label{lemma:covering}
    Let $M$ be globally hyperbolic and $\g,\d:[a,b]\to M$ be two smooth causal curves with compact trajectories $\g,\d$ such that $\g\cap\d=\emptyset$. Then $\g$ admits a finite covering $\g \subseteq \bigcup_{i=1}^n D^+(O_i)$ in terms of future domains of dependence of causally convex, open regions $O_i$ which satisfy $D^+(O_i) \cap \d = \emptyset$ and $\g\cap O_{i+1} \subseteq \g\cap D^+(O_i)$.
\end{lemma}
\begin{proof}
    See Figure~\ref{fig:D-O-covering} on page~\pageref{fig:D-O-covering} for an illustration of the desired covering. Since $M$ is locally compact and Hausdorff, it is regular, implying that every point $p\in M$ has an open neighbourhood $O_p$ such that $O_p\cap\d = \emptyset$. By suitably shrinking $O_p$, we can assume it is causally convex. Since $\d$ is causal, we also have $D^+(O_p)\cap\d = \emptyset$. This gives a covering $\bigcup_{p\in\g} D^+(O_p)$, which admits a finite subcover $\bigcup_{i=1}^n D^+(O_i)$, where the $O_i$ are labelled such that $t^-(O_i) < t^-(O_j) \implies i < j$. Here
    \begin{equation}
        t^-(O) := \inf\{t\in[a,b]:\g(t)\in O\} \text{\quad and\quad} t^+(O) := \sup\{t\in[a,b]:\g(t)\in O\}.
    \end{equation}
    Note that by construction, all $O_i$ intersect $\g$, and that $t^-(D^+(O_i)) = t^-(O_i)$, since $\g$, being causal, cannot be in $D^+(O)\setminus O$ without first passing through $O$ (here we assume w.l.o.g.\ that $\g$ is future-directed). Moreover, since $D^+(O_i)$ is causally convex, $\g$ can enter $D^+(O_i)$ only once, so that $\g\cap\D^+(O)$ is precisely the trajectory of $\g\big|_{(t^-(O),t^+(O))}$. So if $t^-(D^+(O_i))\leq t^-(D^+(O_j)) \leq t^+(D^+(O_j)) \leq t^+(D^+(O_i))$ then we can leave out $D^+(O_j)$ from the cover. Therefore we can assume that $t^+(D^+(O_i)) \leq t^+(D^+(O_{i+1}))$ for all $1\leq i\leq n-1$.
    
    This, together with the fact that all $O_i$ and $D^+(O_i)$ are open, implies that $\g\cap D^+(O_i)\cap O_{i+1} \neq \emptyset$: otherwise, $\g(t^+(D^+(O_i)))$ would not be covered by the covering. Next, $O_{i+1}$, being a causally convex subset of $M$ and thus globally hyperbolic, can be foliated by Cauchy surfaces \cite{Ger70}; hence, find a Cauchy surface $\Sigma$ of $O_{i+1}$ such that $\Sigma$ intersects $\g\cap D^+(O_i)$. Then, noting that $\g\cap D^+(O_i)\cap O_{i+1}$ is open, define $O'_{i+1}$ to be a small causally convex neighbourhood of $\Sigma$ such that $\g\cap O'_{i+1} \subseteq \g\cap D^+(O_i)\cap O_{i+1}$. Finally, set $O'_1 := O_1$. Then $D^+(O_i') = D(\Sigma) = D^+(O_i)$ for all $i$, so $\bigcup_{i=1}^n D^+(O'_i)$ still covers $\g$, and $\{O_i'\}_{i=1}^n$ satisfies the required properties.
\end{proof}